\documentclass[a4paper]{article}
\usepackage[latin1]{inputenc} % ou \usepackage[utf8]{inputenc}
\usepackage[T1]{fontenc} % ou \usepackage[OT1]{fontenc}
\usepackage{RR,RRthemes}
\usepackage{hyperref}
\RRdate{March 2011}
\RRauthor{% les auteurs
Ichrak Amdouni %\thanks{Footnote for first author}%
  \and
Pascale Minet %\thanks{Footnote for second author}%
 \and
Cedric Adjih
}
\authorhead{Amdouni \& Minet \& Adjih}
\RRtitle{Coloriage des noeuds pour des réseaux de capteurs sans fil denses}
\RRetitle{Node coloring for dense wireless sensor networks}
\titlehead{Node coloring}
\RRresume{Le coloriage est utilisée dans les réseaux sans fil pour améliorer l'efficacité des communications, principalement en termes de bande passante, énergie et éventuellement délais de bout-en-bout. Dans ce rapport de recherche, nous définissons le problème de coloriage des noeuds à h sauts, avec h entier positif quelconque. Nous prouvons que le problème de décision associé est NP-complet. Nous présentons un algorithme de coloriage à 3 sauts distribué qui est optimisé pour des réseaux denses: un noeud ne doit plus échanger les couleurs et priorités de ses voisins à 2 sauts. Nous mettons en évidence à travers des résultats de simulation, l'impact de l'assignation de priorité sur le nombre de couleurs obtenu pour un réseau quelconque. Nous nous intéressons ensuite aux réseaux en grille et identifions un motif colorié pouvant être reproduit pour colorier entièrement la grille. Nous montrons comment l'algorithme de coloriage peut utiliser les propriétés de régularité pour obtenir un motif périodique avec le nombre optimal de couleurs. Nous présentons la Méthode des Vecteurs qui donne le nombre optimal de couleurs parmi tous les coloriages périodiques à h sauts. Nous établissons ensuite une borne inférieure et une borne supérieure du nombre de couleurs nécessaires dans un coloriage périodique à h sauts.
}
\RRabstract{Coloring is used in wireless networks to improve communication efficiency, mainly in terms of bandwidth, energy and possibly end-to-end delays. In this research report, we define the h-hop node coloring problem, with h any positive integer. We prove that the associated decision problem is NP-complete. We then present a 3-hop distributed coloring algorithm that is optimized for dense networks:  a node does not need to exchange the priorities and colors of its 2-hop neighbors. Through simulation results, we highlight the impact of priority assignment on the number of colors obtained for any network. We then focus on grids and identify a color pattern that can be reproduced to color the whole grid. We show how the coloring algorithm can use regularity properties to obtain a periodic color pattern with the optimal number of colors. We present the Vector Method that provides the otimal number of colors among all periodic h-hop colorings. We establish lower and upper bounds on the number of colors needed in a periodic h-hop coloring.
}
\RRmotcle{coloriage de graphe, réseaux de capteurs sans fil, MANET, réseaux mobiles ad hoc, efficacité énergétique, réutilisation spatiale, grille, NP-complet, complexité, réseaux denses, motif, borne inférieure, borne supérieure}
\RRkeyword{graph coloring, wireless sensor networks, MANET, mobile ad hoc networks, energy efficiency, spatial reuse, grid, NP-complete, complexity, dense networks, pattern, upper bound, lower bound}
\RRprojet{Hipercom}  % cas d'un seul projet
\RRdomaine{3} % cas du domaine numero 1
\RRthemeProj{hipercom} % theme du projet Hipercom
\RRdomaineProjBis{hipercom} % domaine du projet Hipercom
\RCParis % Paris Rocquencourt
\usepackage{subfigure}
\usepackage{mathenv}
\usepackage{amssymb}%cedric

\graphicspath{{figures/}}

\def\QED{\mbox{\rule[0pt]{1.5ex}{1.5ex}}} %Pascale
\def\proof{\hspace*{-1.4em}{{\itshape Proof: }}} %Pascale
\def\endproof{\hspace*{\fill}~\QED\par\endtrivlist\unskip} %Pascale
\newtheorem{definition}{Definition} %Pascale
\newtheorem{lemma}{Lemma} %Pascale
\newtheorem{property}{Property} %Pascale
\newtheorem{theorem}{Theorem} %Pascale
\newtheorem{equivalence}{Equivalence} %Pascale
\newtheorem{constraint}{Constraint} %Pascale

\begin{document}
\RRNo{7588} 
\makeRR   % cas d'un rapport de recherche
%% \makeRT % cas d'un rapport technique.
%% a partir d'ici, chacun fait comme il le souhaite
\tableofcontents
\newpage

\section{Motivations\label{Motivations}}
%========================================
Coloring has been used in wireless ad hoc and sensor networks to improve communications efficiency by scheduling medium access. Indeed, only nodes that do not interfere are allowed to transmit simultaneously. Hence coloring can be used to schedule node activity. The expected benefits of coloring are threefold:
\begin{enumerate}
\item At the bandwidth level where no bandwidth is lost in collisions, the overhearing and the interferences are reduced. Moreover, the use of the same color by several nodes ensures the spatial reuse of the bandwidth.
\item At the energy level where no energy wasted in collision. Furthermore, nodes can sleep to save energy without loosing messages sent to them because of the schedule based on colors.
\item At the delay level where the end-to-end delays can be optimized by a smart coloring ensuring for instance that any child accesses the medium before its parent in the data gathering tree.\\
\end{enumerate}

However, WSNs (Wireless Sensor Networks) have strong limitations. They have low capacity of storage and computing, low energy especially for battery operated nodes and the network bandwidth is also limited. That is why algorithms supported by WSNs must be of low complexity. More challenging are dense WSNs, where a node cannot maintain its 2-hop neighbors because of memory limitation and a single message cannot contain all the information relative to the 2-hop neighbors of a node. Examples of dense WSNs are given by smart dust where microelectomechanical systems called MEMS can measure temperature, vibration or luminosity. Applications can be  monitoring of building temperature, detection of seismic events, monitoring of pollution, weather prediction for vineyard protection. In this paper, we show how to optimize a coloring algorith for dense WSNs.\\

Concerning coloring algorithms, two types of coloring are distinguished:
node coloring and link coloring. With link coloring, timeslots are assigned per link. Only the transmitter and the receiver, the two nodes of the link are awake, the other nodes can sleep. If the link is lightly loaded, its slot can be underused. Moreover, broadcast communications are not easy: the source must send a copy to each neighbor. On the contrary, with node coloring, the slot is assigned to the transmitter that can use it according to its needs: unicast and/or broadcast transmissions. Hence the slot use is optimized by its owner.\\

The value of $h$ in $h$-hop node coloring depends on the types of communication that must be supported. For instance, broadcast transmissions require 2-hop coloring, whereas unicast transmission with immediate acknowledgement (i.e. the receiver uses the timeslot of the sender to transmit its acknowledgement) requires 3-hop coloring.\\

This paper is organized as follows. First, we define the coloring problem in Section~\ref{Problem} and introduce definitions. We position our work with regard to the state of the art in Section~\ref{StateArt}. In Section~\ref{Complexity}, we prove that the $h$-hop coloring decision problem is NP-complete, for any integer $h>0$. That is why, we propose an heuristic, called SERENA, to color network nodes. The optimization of SERENA, a distributed coloring 3-hop node coloring algorithm, to support dense WSNs is presented in Section~\ref{Optimization}. Performance results, obtained by simulation, are reported in Section~\ref{Performance} for various configurations of wireless sensor networks.  Section~\ref{Theoretical} provides theoretical results related to grid coloring and show how to find the color pattern with the optimal number of colors. These results are used in Section~\ref{Improvement} to deduce a node priority assignment in SERENA for grid topologies that leads to the optimal number of colors. Finally, we conclude in Section~\ref{Conclusion} pointing out future research directions.  

\section {Coloring problem definition \label{Problem}}
%======================================================
Let $G(V,E)$ be a graph representing the network topology. Each vertex $vi \in V$ represents a network node with $i \in [1,n]$, where $n$ is the number of network nodes. For all vertices $v1$ and $v2$ in $V$, the edge $(v1,v2) \in E$ if and only if the two nodes $v1$ and $v2$ are one-hop neighbors.\\ 

There are two types of graph coloring: 
\begin{itemize} 
\item vertex (or node) coloring assigns a color to each vertex of the graph;
\item edge (or link) coloring assigns a color to each vertex of the graph.
\end{itemize}
More precisely, we have the two following definitions:

\begin{definition}
\textbf{One-hop node coloring} of $G$ consists
in (i) assigning each vertex in $V$ a color in such a way that two adjacent vertices have different colors and (ii) using the smallest number of colors.
\end{definition}

\begin{definition}
\textbf{One-hop link coloring} of $G$ consists
in (i) assigning each edge in $E$ a color in such a way that two edges incident to the same vertex have different colors and (ii) using the smallest number of colors.
\end{definition}

We can easily extend one-hop node (respectively link) coloring to $h$-hop node (respectively link) coloring, where $h$ is an integer strictly positive. We introduce the following definitions:

\begin{definition} 
A $h$-hop node coloring is said \textbf{valid} if and only if any two nodes that are $k$-hop neighbors, with $1 \leq k\leq h$ have not the same color.\\
\end{definition}

\begin{definition} 
A $h$-hop link coloring is said \textbf{valid} if and only if any two links that are incident to the same vertex or $k$-hop neighboring vertices, with $1 \leq k\leq h-1$ have not the same color.\\
\end{definition}

\begin{definition} 
A valid $h$-hop node (respectively link) coloring is said \textbf{optimal} if and only if no valid $h$-hop node (respectively link) coloring uses less colors that this coloring.\\
\end{definition}

For simplicity reasons, colors are represented by natural integers, starting with zero.
We can notice that almost all coloring algorithms when applied to a wireless network, make the assumption of an ideal wireless environment.

\begin{definition} 
A wireless environment is said \textbf{ideal} if and only if:
\begin{itemize}
\item Any node has a unique address in the wireless network.
\item Links are symmetric: if node $u$ sees node $v$ as a one-hop neighbor, then conversely node $v$ sees node $u$ as a one-hop neighbor.
\item Links are stable. More precisely, link creation during or after the completion of the coloring algorithm is not taken into account.
\item For any node $u$, any node $w$ that is not in transmission range of $u$ cannot prevent $u$ from correctly receiving a message sent by $v$ of of its one-hop neighbors.\\
\end{itemize}
\end{definition}

It can be easily shown that the breakage of a link does not compromize the validity of a coloring, whereas the creation of a new link can make a coloring no longer valid. The creation of new links can result from node mobility or late arrival of node.\\

To compare the performance of coloring algorithms, two criteria are used:
\begin{itemize}
\item \textbf{the number of colors used}. The optimal number is called chromatic number.
\item \textbf{the number of rounds needed} to color the nodes/links of the graph. By definition, in a round, a node is able to send a message to its neighbors, to receive their message and to process them.
\end{itemize}
The overhead induced by the algorithm is evaluated mainly in terms of:
\begin{itemize}
\item bandwidth: number of messages sent, size of the messages;
\item memory: size of the data structure maintained.\\
\end{itemize}

We now present a brief state of the art dealing with graph coloring and its application to radio networks and wireless sensor networks. 

\section{State of the art\label{StateArt}}
%=========================================
As it can be guessed from Section~\ref{Problem}, coloring has been first introduced in graphs with vertex coloring and edge coloring. 

One-hop vertex coloring has been shown NP-complete in \cite{gar79} for the general case, whereas graphs with maximum vertex degree less than four and bipartite graphs can be colored in polynomial time. The first algorithms proposed were centralized like the greedy algorithm Dsatur (no color backtracking)~\cite{brel79}, where the vertex with the highest number of already colored neighbor vertices is colored first. Later on, 
decentralized ones like Distributed Largest First \cite{hans04} were designed. In this algorithm, each node selects a color. It is allowed to keep its color only if it has the largest degree among its neigbors but also this color does not conflict with the colors already chosen by its neighbors. This algorithm runs in O($\Delta^2 log n$), where $\Delta$ is the largest vertex degree and $n$ the number of vertices. The algorithm given in \cite{kuhn06} proceeds iteratively by reducing the number of colors step-by-step; initially, all nodes have distinct colors. This algorithm runs in O($\Delta^2 log n$) and uses a number of colors close to $\Delta$.\\

Edge coloring problems can be transformed into a vertex version: an edge coloring of a graph is just a vertex coloring of its link graph. Applied to wireless networks, edge coloring has been called link scheduling. For instance,~\cite{Gandham08} obtains a TDMA MAC schedule enabling two-way communication between every pair of neighbors. Edges are colored in such a way that two edges incident on the same node have not the same color. A feasible direction of transmission such that no destination is made unable to receive its message, is searched.\\ 
 
Coloring has then been applied to radio networks to provide a collision-free medium access. The goal is to schedule transmissions in time slots, in such a way that two senders allowed to transmit in the same slot do not interfer. Such a problem is also called broadcast scheduling in~\cite{Rama89}, channel assignment in ~\cite{Krumke00} or slot assignment in ~\cite{interfcolor}. TRAMA,~\cite{rajendran03}, schedules node transmissions by assigning time slots to transmitting nodes and receivers. Only nodes having data to send contend for a slot. The node with the highest priority in its neighborhood up to 2-hop wins the right to transmit in the slot considered. Each node declares in advance its next schedule containing the list of its slots and for each slot its receiver(s). The adaptivity of TRAMA to the traffic rate comes at a price: its complexity. DRAND,~\cite{DRAND}, the coloring algorithm used with the hybrid ZMAC protocol, \cite{ZMAC}, that operates like CSMA under low contention and like TDMA otherwise, assigns slots to nodes in such a way that 1-hop and 2-hop neighbors have different slots. This randomized algorithm has the advantage of not depending on the number of nodes but at the cost of an asymptotic convergence.\\

More recently, coloring algorithms have been designed for WSNs, like FLAMA~\cite{rajendran05}, an adaptation of TRAMA, where the overhead of the algorithm has been considerably reduced. This is obtained by supporting communications of a node only with its parent and its children in the data gathering tree rooted at the sink. TDMA-ASAP \cite{TDMA-ASAP} integrates a coloring algorithm with the medium access. It has been designed for data gathering applications where communications are limited to the data gathering tree. Moreover, this protocol can adapt to various traffic conditions by allowing a node to steal an unused slot to its brother in the tree. In WSNs where energy matters, it is important to use energy efficiently by assigning sensors
with consecutive time slots to reduce the frequency of state
transitions, like~\cite{Ma09}. FlexiTP~\cite{FlexiTP} is a TDMA-based protocol in
which a slot is assigned to one transmitter and one receiver. All other nodes can sleep during this slot.
Slots are assigned such that no nodes that are 1 or 2 hops away transmit in the same slot.
In this protocol, nodes build a tree rooted at the data aggregation sink and run a neighbor discovery phase.
The slot assignment order is given by a deep-first search of the tree.
 A node selects the smallest available slot in its neighborhood up to 2 hops and advertises its schedule.
 
Note that a node does not aggregate data from its children before sending them to its parent.
 Which means that the transitions between idle, transmit, receive activities are
frequent and increase with the network density. This may impact the data gathering delays and the energy consumed by a node.
 This solution does not support immediate acknowledgment.
In \cite{iwcmc08}, we proposed SERENA a node coloring algorithm that increases energy efficiency by avoiding collisions, reducing overhearing, allowing nodes to sleep to save energy and enabling spatial reuse of the bandwidth. This algorithm can support various types of communication (unicast with immediate acknowledgement, broadcast). It can also be optimized for data gathering applications: by scheduling the children before their parent, each parent can aggregate the data gathered from its children before transmitting them to its own parent. In this paper, we show how SERENA overhead can be reduced in dense WSNs. We show new performance results and establish new theoretical results.\\ 

The theoretical performance of coloring algorithms has been studied in the
litterature, in general for 2-hop coloring.
Because minimum 2-hop coloring is NP-hard, the focus has been on evaluating
the performance of \emph{approximation} algorithms, whose objective
is to not find the optimal coloring, but at least to be reasonnably close.
A typical approximation algorithm for coloring is \emph{FirstFit}:
\emph{FirstFit} \cite{Car07} sequentially assigns colors to nodes; it chooses for each
node the first available color.
Depending of the order in which the nodes are colored,
different results (with varied performance) are obtained.
Notice that SERENA is a practical distributed protocol that implements 
an efficient  version of the algorithm \emph{FirstFit} for a $3$-hop coloring, 
with a specific order induced by node priority.
A whole class of results expresses properties related to 
the \emph{worst-case} performance of approximation algorithms: 
they typically prove 
that,  for any input graph of a given family, the coloring obtained by a given 
algorithm  uses at most $\alpha$ times the optimal number of colors. 
Such an algorithm is denoted an $\alpha$-approximation algorithm.\\

Genetic coloring algorithms exist also, like \cite{Chakra04}. 
More generally, we can classify the coloring algorithms according to five criteria:
\begin{itemize}
\item centralized/distributed,
\item deterministic/probabilistic,
\item vertex/edge coloring,
\item types of communication supported,
\item optimized for WSN or not.
\end{itemize}
\begin{table}[!h]
\caption{Classification of coloring algorithms}
\label{AlgoTable}
\begin{tabular}{|c|c|c|c|c|c|}
\hline
& central.& determinist. & vertex & communication  &optimized\\
& distrib. & probabilist. & edge &  &\\
\hline
TRAMA & distrib. & determin. & edge & unicast&\\
\hline
FLAMA & distrib. & determin. & edge & unicast in a tree & data gathering\\
\hline
ZMAC-DRAND & distrib. & random.& vertex& unicast+broadcast& \\
\hline
TDMA-ASAP & central.*& determin. & vertex & unicast in a tree& data gathering\\
\hline
FlexiTP & distrib. & determin. & edge & unicast+broadcast& data gathering\\
\hline
SERENA & distrib. & determin. & vertex & unicast*+broadcast& here*\\
\hline
\end{tabular}

\begin{tabular}{l l l}
\textit{Legend:} &unicast*: &unicast with immediate acknowledgement:\\
%&& the receiver uses the slot of the sender to transmit its acknowledgement.\\
&here*:& optimized in this paper for dense WSNs.\\
&central*:& only the centralized version is described in \cite{TDMA-ASAP}.\\
\end{tabular}
\end{table}

\newpage
\section{Complexity study \label{Complexity}}
%============================================
In this section, we will demonstrate that $h$-hop ($h\ge1$) vertex coloring is a NP-complete problem. This assertation is given by Theorem~\ref{Thcomplexity}: \\

\begin{theorem}
\label{Thcomplexity}
The decision problem of h-hop (h$\ge$1) vertex coloring is NP-complete.
\end{theorem}

It has been proved in~\cite{gar79} that the 1-hop vertex coloring problem is NP-complete. We now prove the NP-completeness for $h \geq2$. Our methodology to prove Theorem~\ref{Thcomplexity} is based on the following steps:\\

\noindent$\bullet$ First, we define the associated decision problem of the $h$-hop vertex coloring of a graph $G$ which is: can this graph $G$ be colored with $k$ colors ($k$ is a positive integer smaller than the vertex number), such that two nodes that are $l$-hop neighbors with $1 \le l \le h$ have not the same color? This problem is called \textbf{$k$-color $h$-hop coloring}.\\

\noindent$\bullet$ Second, we prove the following lemma:
\begin{lemma}
The $k$-color $h$-hop coloring problem is in NP, for $h \ge 2$.
\end{lemma}

\proof Given a $h$-hop coloring of $G$, $h \ge 2$ we can check in polynomial time
($O(n^h)$, where $n$ is the number of nodes) that the coloring
produced by a given $h$-hop algorithm does not assign the same
color to two nodes that are $p$-hop neighbors with $1 \le p \le h$, and that the
total number of colors is $k$. 
\endproof

\vspace*{10pt}
\noindent$\bullet$ Third, we define a reduction $f$ of the $k$-color 1-hop vertex coloring problem that has been shown NP-complete in~\cite{gar79}, to a $k'$-color $h$-hop coloring problem, with $k'$ a positive integer smaller than the nodes number. This reduction should be polynomial in time. Based on this reduction, we then prove the following equivalence: 
\begin{equivalence}
\label{equiv}
A $k'$-color $h$-hop vertex coloring problem has a solution if and only if a $k$-color 1-hop vertex coloring problem has a solution.
\end{equivalence}

In general, to demonstrate that a problem is NP-complete based on another problem that is known to be NP-complete, the required reduction should allow us to show that we can find a solution for the first problem if and only if we can find a solution for the second problem. 
In our case, we should transform a graph $G(V,E)$ to a graph $G'=(V',E')$, and show that finding a $k$-color 1-hop coloring of $G(V,E)$ can lead to find a $k'$-color $h$-hop coloring of $G'(V',E')$ and vice versa, proving Equivalence~\ref{equiv}.\\
Finding a valid $k'$-color $h$-hop coloring of $G'(V',E')$ based on a valid $k$-color 1-hop coloring of $G(V,E)$ requires that for any two nodes $v_1$ and $v_2$ in $G$, the following constraints are met: \\
\begin{constraint}
\label{C1} 
Any two nodes $v_1$ and $v_2$, 1-hop away in $G$ must be at most $h$-hop away in $G'$. 
\end{constraint}
Thus, these two nodes that are assigned different colors by a 1-hop coloring of $G$ are also assigned different colors by a $h$-hop coloring of $G'$.
\begin{constraint}
\label{C2} 
Similarly, any two nodes $v_1$ and $v_2$, 2-hop away in $G$ must be at least $h+1$-hop away in $G'$. 
\end{constraint}

Consequently, the reduction separates any two nodes $v_i$ and $v_j$ of the initial graph $G$ by a set of nodes such that the distance between them in the new graph $G'$ is at most $h$ hops. $V'$ is obtained from $V$ by adding new nodes. The definition of these new nodes depends on $h$ parity. An example is depicted in Figure~\ref{5hopGraph} for $h$=5. 

In order to simplify the determination of $k'$, the number of colors used for the $h$-hop coloring of $G'$, we add to the transformation a new constraint:\\
\begin{constraint}
\label{C3} 
Any two nodes in $V'\setminus V$ must be at most $h$-hop away. Moreover, any two nodes $u \in V$ and $v \in V' \setminus V$ must be at most $h$-hop away.
\end{constraint}
Thus, a $h$-hop coloring of $G'$ cannot reuse a color in $V'\setminus V$. Similarly, no node in $V$ can reuse a color used by a node in $V' \setminus V$.\\

The transformation proceeds as follows, depending on the parity of $h$:
\begin{enumerate}
\item  \textbf{First case: $h$ is odd}: see the example $h=5$ illustrated in Figure~\ref{5hopGraph}. \\
$\bullet$ \textit{Definition of $V'$}\\
In this case, we first define $h'=(h-1)/2$ bijective functions $f_i$ with $i \in[1,h']$:\\
		$\begin{array}{ccccc}
		f_i & : & V & \to & U_i\\ 
		& & v & \mapsto & f_i(v) = u_i\\
		\end{array}$\\

Now, we can define the set $V'$, $V' = V \cup_i U_i\ \cup \{u_0\}$, $\forall i \in [1,h']$, where $u_0$ is a new node introduced to meet constraint C3. Node $u_0$ is a neighbor of all nodes in $U_{h'}$.\\

$\bullet$ \textit{Definition of $E'$}\\
To build the set $E'$, four types of links are introduced. We then have:  $E'= E_1 \cup E_2 \cup E_3 \cup E_4$ where:
\begin{itemize}
\item $E_1 = \{(v,u_1)\ such\ that\ v \in V \ and\ u_1=f_1(v) \in U_1\}$. Thus, each node $v_i$ from the initial graph $G$ is linked to $u_{i1}$, its associated node from the set $U_1$ (see links of type $e_1$ in Figure~\ref{5hopGraph}).
\item $E_2 = \cup_{l \in[1,h'-1]} \{(u_l,u_{l+1})\ such\ that\ u_l \in U_l\ and\ u_{l+1} \in U_{l+1}\ and\ f_l^{-1}(u_l)=f_{l+1}^{-1}(u_{l+1})\}$. Each node $u_{ij}$ from $U_j$ is linked to node $u_{ij+1}$ from $U_{j+1}$ associated with the same node $v \in V$, (see links of type $e_2$). 
\item $E_3 = \{(u_{h'},v_{h'})\ such\ that\ u_{h'}\ and\ v_{h'}\in U_{h'}\ and\ (f_{h'}^{-1}(u_{h'}), f_{h'}^{-1}(v_{h'})) \in E\}$. Two nodes $u_{ih'}$ and $v_{ih'}$ from $U_{h'}$ are linked to each other if their corresponding nodes in $V$ are linked in $E$ (see links of type $e_3$).
\item $E_4 = \{(u,u_0)\ with\ u \in U_{h'}\}$. Finally, the nodes in $U_{h'}$ are linked to the conjunction node $u_0$, which was added to respect the constraint \textbf{C3} (see links of type $e_4$).\\
\end{itemize}

	This construction is polynomial in time. An example of graphs $G$
	and $G'$ with $h=5$ is illustrated in Figure~\ref{5hopGraph}.
	\begin{figure}[!h]
	\begin{center}
	\subfigure[]{\includegraphics[width=0.5\linewidth]{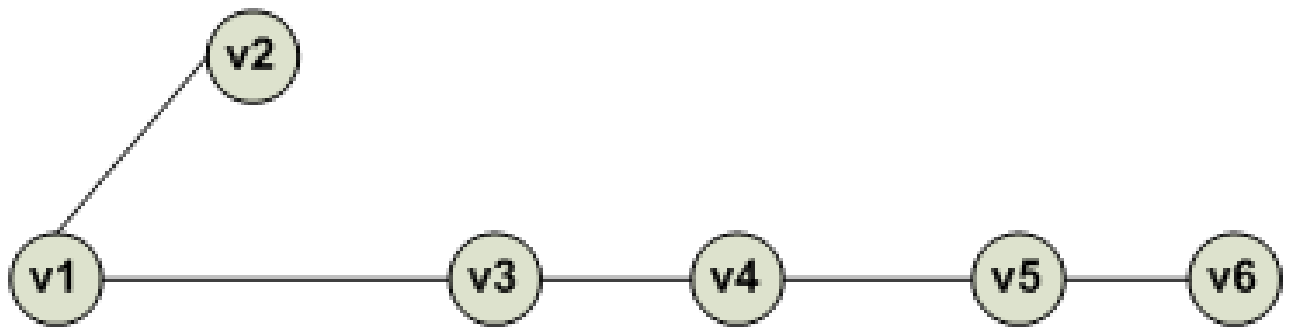}\label{originalGraph}}
	\subfigure[]{\includegraphics[width=0.9\linewidth]{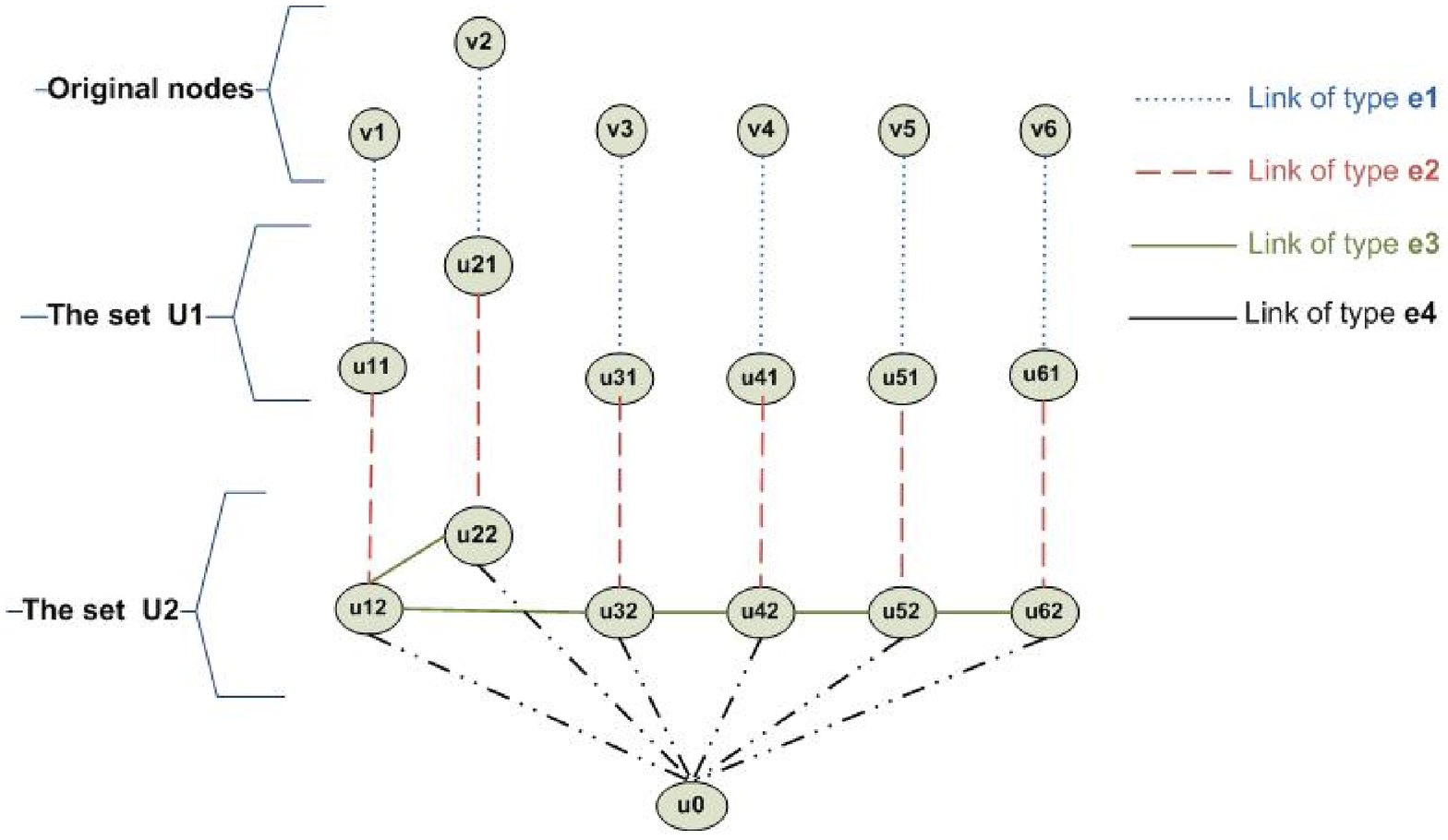}\label{5hop}}
	\caption{Example of: (a) Graph $G$; (b) Transformed graph $G'$ for $h=5$.} \label{5hopGraph} \vspace{-0.3cm}
	\end{center}
	\vspace{-0.2cm}
	\end{figure}

\item \textbf{Second case: $h$ is even}: see the example $h=6$ illustrated in Figure~\ref{6hopGraph}.\\ 

To build the graph $G'$ in the case $h$ is an even number, the same contraints C\ref{C1}, C\ref{C2} and C\ref{ C3} are considered. However, as the number of links to introduce between two nodes in the initial graph $G$ depends on the number of nodes to introduce between them, and thus, on the $h$ parity, we outline some differences in the reduction.

$\bullet$ \textit{Definition of $V'$}\\
In this case, let $h'=h/2$, we first define $h'-1$ bijective functions $f_i$ with $i \in[1,h'-1]$:\\
		$\begin{array}{ccccc}
		f_i & : & V & \to & U_i\\ 
		& & v & \mapsto & f_i(v) = u_i\\
		\end{array}$\\
and the bijective function $f_{h'}$:\\
		$\begin{array}{ccccc}
		f_{h'} & : & E & \to & U_{h'}\\ 
		& & e & \mapsto & f_{h'}(e)= u_{h'}\\
		\end{array}$\\
	
Now, we can define the set $V'$, $V' = V \cup_i U_i$, $\forall i \in [1,h']$.\\

$\bullet$ \textit{Definition of $E'$}\\
To build the set $E'$, five types of links are introduced. We then have: 
$E'= E_1 \cup E_2 \cup E_3 \cup E_4 \cup E_5$ where:\\
\begin{itemize}
\item $E_1 = \{(v,u_1)\ such\ that\ v \in V \ and\ u_1=f_1(v) \in U_1\}$. Thus, each node $v_i$ from the initial graph $G$ is linked to $u_{i1}$, its associated node from the set $U_1$ (see links of type $e_1$ in Figure~\ref{6hopGraph}).
\item $E_2 = \cup_{l \in[1,h'-2]} \{(u_l,u_{l+1})\ such\ that\ u_l \in U_l\ and\ u_{l+1} \in U_{l+1}\ and\ f_l^{-1}(u_l)=f_{l+1}^{-1}(u_{l+1})\}$. Each node $u_{ij}$ from $U_j$ is linked to node $u_{ij+1}$ from $U_{j+1}$ associated with the same node $v \in V$, (see links of type $e_2$). 
\item $E_3 = \{(u_{h'},v_{h'})\ such\ that\ u_{h'}\ and\ v_{h'}\in U_{h'}\ and\ (f_{h'}^{-1}(u_{h'}), f_{h'}^{-1}(v_{h'})) \in E\}$. Two nodes $u_{ih'}$ and $v_{ih'}$ from $U_{h'}$ are linked to each other if their corresponding nodes in $V$ are linked in $E$ (see links of type $e_3$).
\item $E_4 = \{(u_{h'-1},u_{h'}), (u_{h'},v_{h'-1})\ such\ that\ u_{h'-1}\ and\ v_{h'-1} \in U_{h'-1} \ and\ u_{h'} \in U_{h'}\ with\ f_{h'}^{-1}(u_{h'})=(f_{h'-1}^{-1}(u_{h'-1}, f_{h'-1}^{-1}(v_{h'-1}))\}$. In other words, for each couple of nodes $u_{h'-1}$ and $v_{h'-1}$ in $U_{h'-1}$, we associate a node $u_{h'} \in U_{h'-1}$ if and only if $(f_{h'-1}^{-1}(u_{h'-1}), f_{h'-1}^{-1}(v_{h'-1})) \in E$. We then link $u_{h'}$ with $u_{h'-1}$ and $v_{h'-1}$ (see links of type $e_4$).
\item $E_5 = \{(u_i,u_j)\ such\ that\ u_i\ and\ u_j \in U_{h'}\ and \ i \ne j\}$. This means that the nodes in $U_{h'}$ form a complete graph (see links of type $e_5$).
\end{itemize}

This construction is polynomial in time. An example of graphs $G$
and $G'$ with $h=6$ is illustrated in Figure~\ref{6hopGraph}.
\newpage
\begin{figure}[!h]
	\begin{center}
  \subfigure[]{\includegraphics[width=0.5\linewidth]{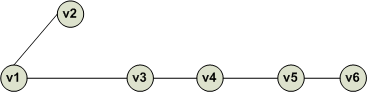}\label{originalGraph}}
  \subfigure[]{\includegraphics[width=0.9\linewidth]{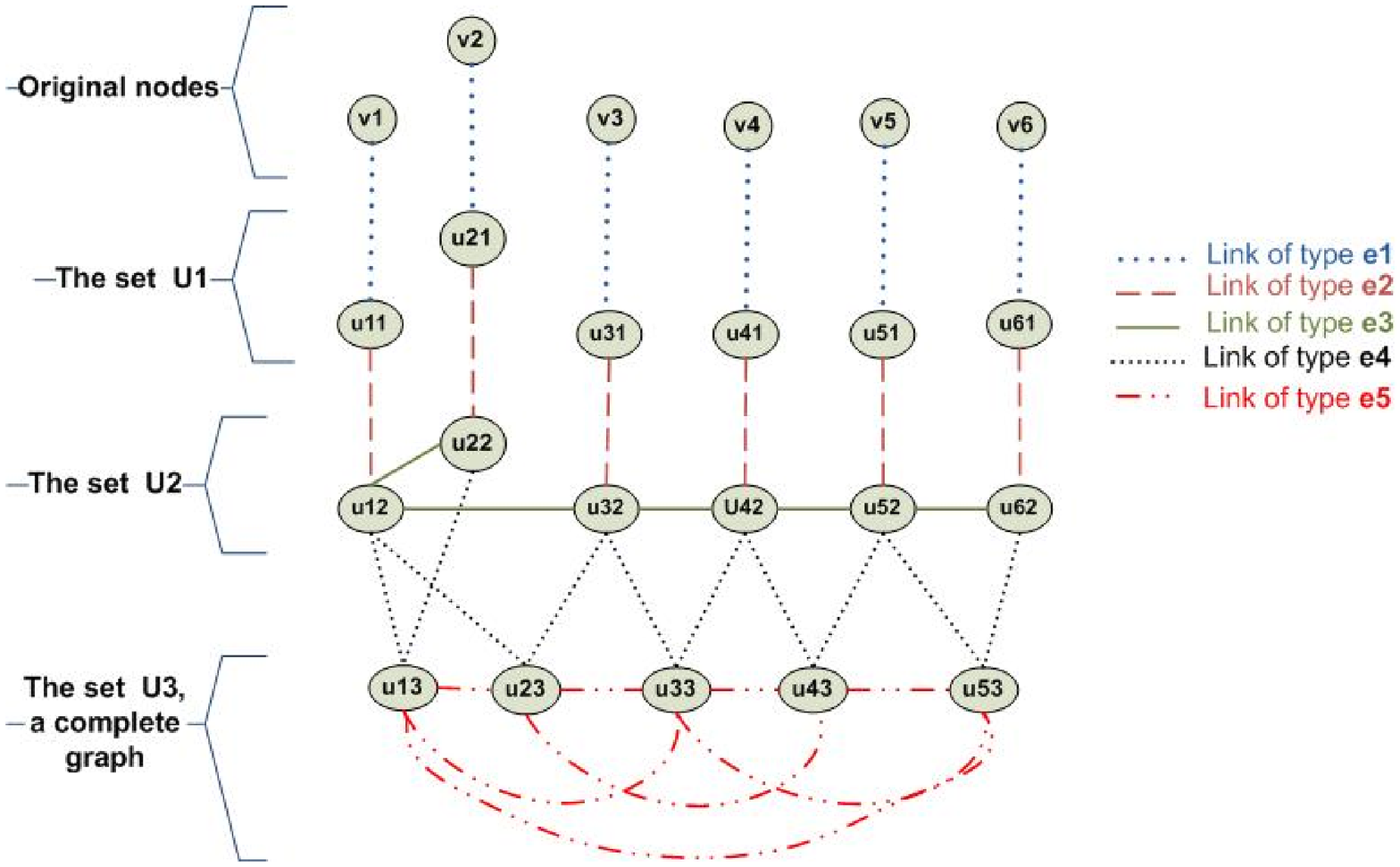}\label{6hop}}
	\caption{Example of: (a) Graph $G$; (b) Transformed graph $G'$ for $h=6$.} \label{6hopGraph} \vspace{-0.3cm}
	\end{center}
	\vspace{-0.2cm}
	\end{figure}
\end{enumerate}

We now show, that the $k'$-color $h$-hop vertex coloring problem, for $h \ge 2$ has a solution if and only if the $k$-color 1-hop vertex coloring problem has a solution. We define the following Lemma:

\begin{lemma}
\label{lemmah-1}
All nodes in $G'\setminus G$ are at most $(h-1)$-hop neighbors.
\end{lemma}

\vspace{-6pt}
\proof By construction of $G'$.
\endproof

\vspace{+12pt}
\begin{lemma}
\label{lemmaG'} 
To perform a $h$-hop coloring of the graph $G'$, the number of colors taken by nodes in $V'\setminus V$ is equal to $m$ with $m$ is equal to $(h'\times n)+1$ if $h$ is an odd number, and $(h'\times n)-1$ if $h$ is an even number, where $n$ is the number of nodes in $G$.
\end{lemma}

\vspace{-6pt} \proof From Lemma~\ref{lemmah-1}, all nodes in $V'\setminus V$ are at most 
$(h-1)$-hop neighbors. Hence, no color can be reused with
$h$-hop coloring ($h \ge 2$) of $G'$. By construction of $G'$, the number of these
nodes is equal to $(h'\cdot n)+1$ if $h$ is an odd number, and $(h'\cdot n) -1$ if $h$ is an even number.
\endproof

\vspace{+12pt}
\begin{lemma}
\label{lemmaV'}
Any color used for a node in $V$ by a $h$-hop coloring of $G'=(V',E')$ cannot be used by any node in $V' \setminus V$.
\end{lemma}

\vspace{-6pt} \proof Let us consider any node $u \in  V' \setminus V$ and any node $v \in V$. 
Let $d(v,u)$ be the number of hops between $v$ and $u$. By construction, $d(v,u)=d(v, f_1(v)) + d(f_1(v),u)$. From Lemma~\ref{lemmah-1}, $d(f_1(v),u) \leq h-1$ and since $f_1(v)$ is a neighbor of $v$, we get $d(u,v) \leq h$.
Hence, $u$ and $v$ must use different colors with $h$-hop coloring of $G'$ for $h \ge 2$.
\endproof

\vspace{10pt}
\noindent To complete the proof of Theorem~\ref{Thcomplexity}, we now prove the following Lemma
\begin{lemma}
$G(V,E)$ has a one-hop
coloring with $k$ colors if and only if $G'(V',E')$ has a $h$-hop
coloring with $k'$ colors, with $h \ge 2$.
\end{lemma}

\proof Given a one-hop coloring of $G$ with $k$ colors, we want
to show that there exists a $h$-hop coloring of $G'$ with $k'$
colors as follows. According to Lemma~\ref{lemmaG'}, this $h$-hop coloring will use $k$ colors for
nodes in $V$ and $m$ colors for nodes in $
V'\setminus V$ with $m$ is equal to $(h'\cdot n)+1$ if $h$ is an odd number, and $(h'\cdot n)-1$ if $h$ is an even number. From Lemma~\ref{lemmaV'}, colors used in $V$ cannot be
reused in $V' \setminus V$. It follows that there exists a $h$-hop
coloring of $G'$ with exactly $k'=k+m$ colors.

Now, let us assume that we have a $h$-hop coloring of
$G'$ with $k'$ colors and we want to show that we can find a one-hop
coloring of $G$ with $k$ colors. From Lemma~\ref{lemmaG'}, $m$ colors are
needed for $h$-hop coloring of nodes in $V'\setminus V$. From
Lemma~\ref{lemmaV'}, colors used in $V$ cannot be reused in $V'\setminus V$.
Hence, $k'-m$ colors are used to color the nodes in $V$. Moreover, since
any two nodes $v1$ and $v2$ in $V$ that are one-hop neighbors in $G$
are $h$-hop neighbors in $G'$, by construction of $G'$, we deduce
that no two one-hop neighbors in $G$ use the same color. Hence, we
can find a valid one-hop coloring of $G$ with $k=k'- m$ colors.
\endproof

\newpage
\section{Coloring optimized for dense networks \label{Optimization}}
%======================================================================
The goal of this section is to make possible the use of the coloring algorithm in dense wireless sensor networks. We show how to reduce the overhead in terms of memory required to store the data maintained by each node and bandwith used by exchanging messages between neighbors. Of course this overhead reduction must not decrease the performance of the coloring algorithm: the number of colors and the number of rounds needed to color all network nodes must be the same. First, we give the basic principles of 3-hop coloring.

\subsection{Basic principles of 3-hop node coloring}
%--------------------------------------------
As previously said, 3-hop node coloring is necessary to support unicast transmissions with immediate acknowledgement in case of general communications, where any node is likely to exchange information with any neighbor node.

\noindent In SERENA, any node $u$ proceeds as follows to color itself:
\begin{enumerate}
\item Node $u$ characterizes the set $\mathcal{N}(u)$ of nodes that cannot have the same color as itself. This set depends on the type of:
\begin{itemize}
\item \textit{communications supported}: unicast and/or broadcast;
\item \textit{application}: general where any node is likely to exchange information with any neighbor node or on the contrary tree type where a node exchanges information only with its parent and its children in the data gathering tree;
\item \textit{acknowledgement for unicast transmissions}: immediate or deferred.
\end{itemize}
In our case, the set $\mathcal{N}(u)$ is the set of neighbors up to 3-hop.

\item Node $u$ computes its priority. This priority consists of two components: the most important one, denoted $prio$, is equal to the number of nodes up to two-hop from $u$. The second one denotes the address of the node. By definition, node $u$ is said to have a priority higher than node $v$ if and only if:
\begin{itemize}
\item either $prio(u) > prio(v)$;
\item or $prio(u)= prio(v)$ and $adress(u) < adress(v)$.
\end{itemize}

\item Node $u$ applies the two following rules:
\begin{itemize}
\item \textbf{Rule R1}: Node $u$ colors itself if and only if it has a priority strictly higher than any uncolored node in $\mathcal{N}(u)$.
\item \textbf{Rule R2}: To color itself, node $u$ takes the smallest color unused in $\mathcal{N}(u)$.
\end{itemize}
\end{enumerate}

\subsection{Motivations and optimization principles}
%--------------------------------------------------
\noindent This distributed coloring algorithm proceeds by iterations or rounds, where nodes exchange their $Color$ message. In a naive implementation, this message would include the priority and color of 1) the node $u$ itself, 2) its 1-hop neighbors in $\mathcal{N}(u)$, as well as 3) its 2-hop neighbors in $\mathcal{N}(u)$. The data locally maintained by any wireless sensor would include these data as well as the priority and color of any $3$-hop neighbor. It is well known that the average number of nodes in the neighborhood up to 2-hop is in O($2^2 \cdot \Pi \cdot density \cdot R^2)$, where $density$ stands for the number of nodes per square meter and $R$ is the transmission range. Such an overhead can be unacceptable for wireless sensors with limited storage and processing capabilities as well as low energy. Dense networks with limited bandwidth, low energy and a short MAC frame size become challenging for a coloring algorithm. 
 That is why, we propose in this paper an optimization of the coloring algorithm reducing the size of $Color$ messages exchanged, the size of data structures maintained, while keeping a low complexity. We also show that this overhead reduction does not increase the convergence time of the coloring algorithm. More precisely, we have the following property:
\begin{property}
\label{PropertyP1}
Let $u$ be any node coloring itself at round $r>0$ and $v \in \mathcal{N}(u)$ be the next node to color itself. Node $v$ colors itself at round $r+h$, where $v$ is a $h$-hop neighbor of $u$, with $1 \le h \le 3$.\\
\end{property}

\noindent The optimization principles are based on the following remarks:
\begin{itemize}
\item It is necessary that any node $u$ knows the highest priority of its uncolored neighbors up to 3-hop in order to apply Rule R1. Furthermore, node $u$ must send information concerning itself, its 1-hop and 2-hop neighors to let its one-hop neighbors know information about their 1-hop, 2-hop and 3-hop neighbors.
Hence, node $u$ must send its priority, the highest priority of its 1-hop neighbors as well as the highest priority of its 2-hop neighbors. However, this would not suffice:  Property~\ref{PropertyP1} would be violated. Node $v$, 2-hop away from node $u$ colored at round $r$ would not know at round $r+2$ that it has the highest priority. Hence, the TWO highest priorities at respectively 1-hop and 2-hop must be maintained and sent. The highest priority at 3-hop is locally computed.
\item Similarly for the color, node $u$ must know the colors already used in its neighborhood up to 3-hop. However, it does not matter $u$ to know which node has which color, but only which colors are taken at 1-hop, 2-hop and 3-hop respectively. That is why, we use the fields $color-bitmap1$, $color-bitmap2$ and $color-bitmap3$ for the colors used at 1-hop, 2-hop and 3-hop respectively.
\end{itemize}

\newpage
\subsection{Optimized coloring algorithm} 
%----------------------------------------
\subsubsection{The $Color$ message}
%%%
The format of message Color is depicted hereafter:
%\newpage
\begin{verbatim}
 0                   1                   2                   3
 0 1 2 3 4 5 6 7 8 9 0 1 2 3 4 5 6 7 8 9 0 1 2 3 4 5 6 7 8 9 0 1
 +-+-+-+-+-+-+-+-+-+-+-+-+-+-+-+-+-+-+-+-+-+-+-+-+-+-+-+-+-+-+-+-+
 |    Type       |     Originator Address        |     color     |
 +-+-+-+-+-+-+-+-+-+-+-+-+-+-+-+-+-+-+-+-+-+-+-+-+-+-+-+-+-+-+-+-+
 |    prio       |  size_max2_   |  max2_prio1
 |               |       prio1   |
 +-+-+-+-+-+-+-+-+-+-+-+-+-+-+-+-+-+-+-+-+-+-+-+-+-+-+-+-+-+-+-+-+
                 |      ....
 +-+-+-+-+-+-+-+-+-+-+-+-+-+-+-+-+-+-+-+-+-+-+-+-+-+-+-+-+-+-+-+-+
 |   size_max2_  |       max2_prio2                              |
 |   prio2       |                                               |
 +-+-+-+-+-+-+-+-+-+-+-+-+-+-+-+-+-+-+-+-+-+-+-+-+-+-+-+-+-+-+-+-+
                ....                             |  size_bitmap1 |
 +-+-+-+-+-+-+-+-+-+-+-+-+-+-+-+-+-+-+-+-+-+-+-+-+-+-+-+-+-+-+-+-+
 |color_bitmap1    ...
 +-+-+-+-+-+-+-+-+-+-+-+-+-+-+-+-+-+-+-+-+-+-+-+-+-+-+-+-+-+-+-+-+
 |  size_bitmap2 | color_bitmap2 ...
 +-+-+-+-+-+-+-+-+-+-+-+-+-+-+-+-+-+-+-+-+-+-+-+-+-+-+-+-+-+-+-+-+
\end{verbatim}

\noindent By definition, for any node $u$, $prio(u)$ is equal to the sum of the numbers of 1-hop neighbors of its 1-hop neighbors. This computation is done during the initialization of the coloring algorithm. We also define $max2\_prio1(u)$ as:
\begin{itemize}
\item the two highest priorities of the uncolored 1-hop neighbors of $u$, if two such nodes exist;
\item the priority of the only one uncolored 1-hop neighbor, if only one such node exists;
\item empty, if none exists.\\
\end{itemize}

\noindent Similarly, we define $max2\_prio2(u)$ as the two highest priorities of the uncolored 1-hop neighbors of the 1-hop neighbors of $u$, if they exist. The variable $max\_prio3(u)$ is defined as the highest priority of the uncolored two-hop neighbors of the 1-hop neighbors of $u$. The computation of $max2\_prio1(u)$, $max2\_prio2(u)$ and $max\_prio3(u)$ is done from the $Col$ messages received during the current round. The values computed are inserted in the $Color$ message sent by node $u$.\\

\noindent Notice that the size of the $Color$ message is variable for two reasons. Since $max2\_prio1$ (resp. $max2\_prio2$) can contain 0, 1, or 2 priority values, its size is given in the field $size\_max2\_prio1$ (resp. $size\_max2\_prio2$).
Furthermore, the size of the bitmaps used at 1-hop and 2-hop respectively depends on network topology. We introduce the fields $size\_bitmap1$ and $size\_bitmap2$ to contain these sizes.
 
\subsubsection{Processing}

\noindent With this optimization, Rule R1 becomes:  Any node $u$ colors itself if and only if $Priority(u) = max \{max2\_prio1(u), max2\_prio2(u), max\_prio3(u)\}$.\\

\noindent Rule R2 becomes: Node $u$ selects the smallest color unused in  $color\_bitmap1(u)$ $\cup$ $color\_bitmap2(u)$ $\cup$ $color\_bitmap3(u)$.\\

\noindent Notice that this color should also not be used by heard nodes (nodes with which there is no symmetric link). This, in order to avoid color conflicts.\\

\noindent If at a round $r$>1 of the coloring algorithm, node $u$ does not receive a message from its 1-hop neighbor $v$, it uses the information received from $v$ at round $r-1$.\\

\noindent The coloring algorithm ends when node $u$ as well as all its 1-hop, 2-hop and 3-hop neighbors are colored.\\

\noindent When a node computes $max2\_prio1$, $max2\_prio1$ and $max\_prio3$ from the values received in the $Color$ messages, it discards any value corresponding to an already colored node.

\newpage
\section{Impact of priority assignment on grid coloring \label{Performance}}
%=================================================

We now consider grid topologies. Such topologies exhibit regularity properties. We want to see if SERENA is able to preserve such regularity in the coloring. Hence, we evaluate the performance of SERENA coloring algorithm in various grid topologies by simulation. Does the coloring keep some regularity of the grid topology? Can we find on the colored topology a color pattern that is reproduced several times? 
As said in Section~\ref{Problem}, performance of SERENA is evaluated by the number of colors and the number of rounds needed to color all network nodes.\\

In all the grids considered, we assume a transmission range higher than or equal to the grid step in order to get the radio connectivity. For simplicity reasons, the transmission range $R$ is expressed as a function of the grid step, that is considered as the unit. Hence, $R \geq 1$. Moreover, we assume an ideal environment where any node $u$ is able to communicate via a symmetric link to any node $v$ such that $d(u,v) \leq R$, where $d(u,v)$ denotes the euclidian distance from $u$ to $v$.

\subsection{Impact of node number}
%----------------------------------------

We first assume a radio range equal to the grid step. In other words, the neighbors at the grid sense are also neighbors in the communication sense. We first consider a 10x10 grid with 100 nodes. Assuming a priority assignment as described in Section~\ref{Optimization}, SERENA obtains 13 colors as shown in Table~\ref{TableSerenaRandom}.  %, where the '*' symbol highlights the optimality of the number of colors used. 
If now, we consider a 30x30 grid with 900 nodes, SERENA gets 16 colors for the same priority assignment as previously. Why do not we have the same number of colors? The reason is given by the node priority assignment. All nodes that are not border ones have the same number of 1-hop and 2-hop neighbors respectively. Hence, ties are broken by means of node address. A random node address assignment leads to a non-optimal number of colors. Table~\ref{TableSerenaRandom} provides the number of colors and rounds obtained for different values of the transmission range. We can conclude that the number of colors strongly depends on the density of nodes and weakly on the number of nodes. The number of rounds depends on the number of nodes.
\newpage
\begin{table}[!h]
\caption{Number of colors obtained by SERENA with a random priority assignment for various transmission ranges and grid sizes.}
\label{TableSerenaRandom}
\begin{tabular}{|c|c|c|c|}
\hline
Radio range & grid size & colors &rounds\\
\hline
1 & 10x10 & 13 & 65\\
  \cline{2-4}
  & 20x20 & 14 & 86\\
  \cline{2-4}
  & 30x30 & 16 & 97\\
\hline
1.5 & 10x10 & 26 & 109\\
  \cline{2-4}
  & 20x20 & 28 & 157\\
  \cline{2-4}
  & 30x30 & 28 & 179\\
\hline
2 & 10x10 & 36 & 171\\
  \cline{2-4}
  & 20x20 & 41 & 257\\
  \cline{2-4}
  & 30x30 & 44 & 298\\
\hline 
\end{tabular} 
\end{table}

\subsection{Impact of priority assignment}
%----------------------------------------
Another address assignment produces another coloring using 8 colors for the 10x10 grid and a radio range of 1, as shown in Table~\ref{TableSerenaVariousPrioNodeNb}, where the '*' symbol highlights the optimality of the number of colors used. 
\begin{table}[!h]
\caption{Number of colors obtained by SERENA for various transmission ranges, grid sizes and priority assignments.}
\label{TableSerenaVariousPrioNodeNb}
\begin{tabular}{|c|c|c|c|}
\hline
Radio range & grid size & priority assignment & colors\\
\hline
1 & 10x10 & line & 8*\\
          \cline{3-4}
          && column & 8*\\
          \cline{3-4}
          && diagonal &8*\\
          \cline{3-4}
          && distance to origin &8*\\
    \cline{2-4}
  & 20x20 & line & 15\\
           \cline{3-4}
          && column & 15\\
          \cline{3-4}
          && diagonal &8*\\
          \cline{3-4}
          && distance to origin &8*\\
\hline
2 & 10x10 & line & 30\\
          \cline{3-4}
          && column & 30\\
          \cline{3-4}
          && diagonal &28\\
          \cline{3-4}
          && distance to origin &30\\
  \cline{2-4}
  & 20x20 & line & 33\\
          \cline{3-4}
          && column & 33\\
          \cline{3-4}
          && diagonal &29\\
          \cline{3-4}
          && distance to origin &30\\
 \hline 
\end{tabular} 
\end{table}

The question is can we find a priority assignment in grid topologies such that the coloring does not depend on node number but only on radio range? Moreover, can we find a color pattern that can tile the whole topology?

\newpage
\section{Theoretical results in grid topologies\label{Theoretical}}
%================================================
The goal of this section is to determine the optimal color number for the 3-hop coloring of grids with various transmission ranges. 
In this paper, we only study grid colorings that reproduce periodically a color pattern. As a consequence, the optimality of a coloring obtained is only true in the class of periodic colorings.

\subsection{Notation and definitions}
%------------------------------------
We adopt the following notation and definitions:
Let $R$ denote the transmission range.
\begin{definition}
For any integer $g>0$, for any node $u$ that is not a border node in the grid, its g-square is defined as the square centered at $u$, with a square side equal to 2g. The g-square contains exactly $8g$ nodes.
\end{definition}
\begin{definition}
For any integer $g>0$, for any node $u$ that is not a border node in the grid, its g-diamond is defined as the diamond centered at $u$, with a diagonal length equal to $2g$. The g-diamond contains exactly $4g$ nodes.
\end{definition}
\begin{definition}
\label{basicpatterndef}
A basic color pattern is the smallest color pattern that can be used to periodically tile the whole grid.
\end{definition}
\begin{definition}
\label{optimalpatterndef}
A basic color pattern is said optimal if and only if it generates an optimal periodic coloring of the grid.
\end{definition}

\subsection{Properties independent of the transmission range}
%--------------------------------------------------------------
We can now give properties that do not depend on the transmission range value.
\begin{property}
Any color permutation of an optimal basic pattern is still valid and optimal.
\end{property}
\proof
With the color permutation, no two nodes that are 1-hop, 2-hop or 3-hop neighbors have the same color. Hence, the permuted coloring obtained is still valid. The permutation keeps unchanged the number of colors. Hence the coloring is still optimal.
\endproof
\begin{property}
\label{Pcolgrid}
Given an optimal color pattern of any grid and the color at node of coordinates (0,0), we can build a 3-hop coloring of a grid topology based on this pattern such that the color of node $(0,0)$ is the given color.
\end{property}
\proof
The coloring of the grid is obtained by setting the optimal color pattern in such a way that the color of node $(0,0)$ is the given color. The pattern is then reproduced to tile the whole topology.
\endproof

\begin{property}
Knowing an optimal color pattern of its grid and the color at node of coordinates $(0,0)$, each node can locally determine its own color based on its coordinates $(x,y)$. The 3-hop coloring obtained for the grid is optimal in terms of colors and rounds.
\end{property}
\proof The 3-hop coloring obtained for the grid only requires each node to know the color of node $(0,0)$, its coordinates in the grid and the optimal pattern to apply. Hence, it is optimal in terms of colors and rounds.
\endproof

\subsection{Optimal coloring for various transmission ranges}
%-----------------------------------------------------------
We now prove the optimal coloring of grids for various transmission ranges: $R=1$, $R=1.5$ and $R=2$.

\subsubsection{Transmission range = grid step}
%---------
In this section, we assume a transmission range equal to the grid step.

\begin{definition}
In a grid with a transmission range equal to the grid step, a non-border node is a node in the grid that has exactly 4 1-hop neighbors, 8 2-hop neighbors and 12 3-hop neighbors. 
\end{definition}
For any non-border node $u$, its neighborhood up to 3-hop, $\mathcal{N}(u)$ is illustrated in Figure~\ref{Neighbor3hopFig}, where nodes $a, b, c, d$ denote the 1-hop neighbors, nodes $e$ to $l$ the 2-hop neighbors and nodes $m$ to $y$ the 3-hop neighbors.
\begin{figure}[!h]
\begin{verbatim}
      r
    s h q
  t i b g p
v j c u a f o
  w k d e n
    x l m
      y
\end{verbatim}
\caption{Neighborhood up to 3-hop of node $u$, R=grid step.}
\label{Neighbor3hopFig}
\end{figure}

%\newpage
\textit{a) Optimal color pattern in a grid}\\
%------------------------------------------
$ $\\
Our methodology consists in providing a valid coloring of any non-border node $u$ of the grid as well as all nodes up to 3-hop from $u$. Then we use the coloring obtained to color the whole grid.
\begin{theorem}
The optimal 3-hop coloring of a grid topology with a transmission range equal to the grid step requires exactly 8 colors. An optimal color pattern is given in Figure~\ref{NeighborPatternFig}.
\begin{figure}[!h]
\begin{verbatim}
      4
    5 8 3
  4 7 2 6 4
5 8 3 1 5 8 3
  2 6 4 7 2
    5 8 3
      2
\end{verbatim}
\caption{Coloring of node $u$ and its neighborhood up to 3-hop, R= grid step}
\label{NeighborPatternFig}
\end{figure}
\end{theorem}
\proof
Let $u$ be any non-border node. The proof is done in three steps:
\begin{enumerate}
\item \textit{First step: At least 8 colors are needed to color node $u$ and $\mathcal{N}(u)$.}\\
\textit{First substep:} node $u$ itself requires a color, denoted $1$ for simplicity reasons, that is not used by any other node in $\mathcal{N}(u)$.\\ 
\textit{Second substep:} any 1-hop neighbor of $u$ is 2-hop neighbor of any other 1-hop neighbor of $u$. It follows that any 1-hop neighbor requires a distinct color. We then get 4 colors, denoted $2$ to $5$, for these 1-hop neighbors.\\ 
\textit{Third substep:} we now consider the 2-hop neighbors. Notice that they are at most 3-hop away of any 1-hop neighbor of $u$. Hence, they cannot reuse the colors $2$ to $5$. Moreover, nodes $g$ and $k$ that are 4-hop away can use the same additional color $6$. Similarly, nodes $e$ and $i$, 4-hop away but 2-hop away from color $6$, use an additional color $7$. The remaining 2-hop neighbors, $f$, $h$, $j$ and $l$, 4-hop away, can use the same color. This color would be at most 3-hop away from any already used color. Hence, an additional color $8$ is needed. Hence, at least 8 colors are needed.
\item \textit{Second step: We build a valid coloring of $u$ and $\mathcal{N}(u)$ with 8 colors.} Each 3-hop neighbor of $u$ is 4-hop away from either two or three 1-hop neighbors of $u$. Hence, it can reuse their colors. We consider first 3-hop nodes that have the least color choice, namely nodes like $p$ and $q$. Each of them has the same choice between two colors of two 1-hop neighbors of $u$, namely $c$ and $d$, from which they are 4-hop away. We color first $p$ with an already used color, $4$ for instance. We then have only one possibility for node $q$, color $3$. We proceed similarly for nodes $s$ and $t$ with colors $5$ and $4$, then for nodes $w$, $x$ with colors $2$ and $5$ and finally for nodes $m$ and $n$ with colors $3$ and $2$. Now, we consider the remaining three-hop neighbors of $u$, namely nodes $o$, $r$, $v$ and $y$. At the beginning of this step, these nodes had 3 choices (the color of 3 1-hop neighbors of $u$), but as their 2-hop neighbors are now colored, only one choice remains valid: we take this remaining color. Hence, no additional color is introduced. We have used exactly eight colors to color any node $u$ and its neighborhood up to 3-hop, as depicted in Figure~\ref{NeighborPatternFig}.
\item \textit{Third step: This coloring can be regularly reproduced to constitute a valid coloring of the grid.} We consider the origin at node $u$. 
Observing the coloring depicted in Figure~\ref{NeighborPatternFig}, we notice that any color found at coordinates $(x,y)$ is also found:
\begin{itemize}
\item in the same line, at nodes $(x+4,y)$ and $(x-4,y)$,
\item in the same column, at nodes $(x,y+4)$ and $(x,y-4)$,
\item in the same diagonal, at nodes $(x+2, y-2)$, $(x+2,y+2)$, $(x-2,y+2)$ and $(x-2, y-2)$.
\end{itemize}
We then get a coloring of the grid with exactly 8 colors. We prove that this coloring is valid by checking that any color $1$ to $8$ is reused neither 1-hop, nor 2-hop, nor 3-hop away. It follows that this coloring is valid. Hence, an optimal coloring requires exactly 8 colors.
\end{enumerate}
\endproof

\begin{property}
A basic color pattern of the grid with a transmission range equal to the grid step is given by: 
\begin{verbatim}
    7 2 6 4
    3 1 5 8
\end{verbatim}
\end{property}
\proof
We can extract from the coloring of the grid a basic color pattern containing exactly eight colors. This pattern is periodically reproduced to generate the coloring of the grid. Each color is reproduced according to the rules given previously. 
\endproof
$ $\\
\textit{b) From an optimal color pattern to an optimal 3-hop coloring of a grid}\\
%--------------------------------------------------------------------------------
$ $\\
We now show that we can tile the grid topology by reproducing the color pattern previously found. More precisely, a node of coordinates(x,y) in the grid can deduce its color in an optimal 3-hop grid coloring defined as said in property~\ref{Pcolgrid}.
\begin{theorem}
Let $\mathcal{P}$ be an optimal color pattern for any grid and $c0$ the color of node of coordinates (0,0) in the grid, with $1 \leq c0 \leq 8$. The color of any point with coordinates $(x,y)$ in the grid is given by the color of coordinates $(x',y')$ in the pattern $\mathcal{P}$, where the point of color $c0$ is chosen as the origin and with $x'= x\ modulo\ 4$ and $y'=y \ modulo\ 2$.
\end{theorem}
\proof
We position the optimal color pattern in such a way that $c0$ is the color of node of coordinates (0,0) in the grid, we then reproduce: 
\begin{itemize}
\item in each line, the color of $(x,y)$ at nodes $(x+4,y)$ and $(x-4,y)$,
\item in each column, the color of $(x,y)$ at nodes $(x,y+4)$ and $(x,y-4)$,
\item in each diagonal, the color of $(x,y)$ at nodes $(x+2, y-2)$, $(x+2,y+2)$, $(x-2,y+2)$ and $(x-2, y-2)$.
\end{itemize}
Hence, the theorem.
\endproof

\subsubsection{Transmission range = 1.5 x grid step}
%----------------------------------------

Assuming a transmission range equal to 1.5 grid step, we notice that each node $u$ that is not a border one has exactly at the communication sense:
\begin{itemize}
\item 8 1-hop neighbors: such nodes belong to the 1-square. 
\item 16 2-hop neighbors: such nodes belong to the 2-square.
\item 24 3-hop neighbors: such nodes belong to the 3-square.
\end{itemize}
For any non-border node $u$, its neighborhood up to 3-hop, $\mathcal{N}(u)$ is illustrated in Figure~\ref{Neighbor3hopwith1.5Fig}, where nodes $a, b, c, d, e, f, g, h$ denote the 1-hop neighbors, nodes $i$ to $y$ the 2-hop neighbors and nodes $z$ to $w'$ the 3-hop neighbors.
\begin{figure}[!h]
\begin{verbatim}
k' j' i' h' g' f' e'
l' p  o  n  m  l  d'
m' q  d  c  b  k  c'
n' r  e  u  a  j  b'
o' s  f  g  h  i  a'
p' t  v  w  x  y  z
q' r' s' t' u' v' w'     
\end{verbatim}
\caption{Neighborhood up to 3-hop of node $u$, $R=1.5$.}
\label{Neighbor3hopwith1.5Fig}
\end{figure}

\begin{theorem}
An optimal coloring of a grid with a transmission range equal to 1.5 times the grid unit needs exactly 16 colors. An example of optimal coloring is given by the following color pattern:
\begin{figure}[!h]
\begin{verbatim}
 9 16  7  8  9 16  7
13 10 11 12 13 10 11
 3 14  5  4  3 14  5
 2 15  6  1  2 15  6
 9 16  7  8  9 16  7
13 10 11 12 13 10 11
 3 14  5  4  3 14  5   
\end{verbatim}
\caption{Coloring of $u$ and its neighborhood up to 3-hop, R=1.5.}
\label{ColorNeighbor3hopwith1.5Fig}
\end{figure}
\end{theorem}
\proof
Let $u$ be any non-border node. The proof is done in three steps:
\begin{enumerate}
\item \textit{First step: At least 16 colors are needed to color node $u$ and $\mathcal{N}(u)$.}\\
\textit{First substep:} node $u$ itself requires a color, denoted $1$ for simplicity reasons, that is not used by any other node in $\mathcal{N}(u)$.\\ 
\textit{Second substep:} any 1-hop neighbor of $u$ is 2-hop neighbor of any other 1-hop neighbor of $u$. It follows that any 1-hop neighbor requires a distinct color. We then get 8 colors, denoted $2$ to $9$, for these 1-hop neighbors.\\ 
\textit{Third substep:} we now consider the 2-hop neighbors. Notice that they are at most 3-hop away of any 1-hop neighbor of $u$. Hence, they cannot reuse the colors $2$ to $9$. Moreover, nodes $p$, $l$, $y$ and $t$ that constitute the four vertices of the 2-square, are 4-hop away, they can use the same additional color $10$. If we consider the upper side of this square, nodes $o$, $n$ and $m$ are at most 2-hop away, they cannot reuse the same color. Three additional colors are needed: colors 11, 12 and 13. We can now color the lower side of this 2-square by reproducing the colors used on the upper side, 4-hop away. We now consider the left side of this 2-square occupied by nodes $q$, $r$ and $s$. These nodes are at most 2-hop away, they cannot reuse the same color. Three additional colors are needed: colors 14, 15 and 16. We can now color the right side of this 2-square by reproducing the colors used on the left side, 4-hop away. Hence, at least 16 colors are needed.
\item \textit{Second step: We build a valid coloring of $u$ and $\mathcal{N}(u)$ with 16 colors.} Concerning the 3-hop neighbors, they occupy the 3-square. We color the upper line of this square by copying the line 4-hop lower. Similarly, the lower line of this square is colored by copying the colors used by the line 4-hop higher. We proceed similarly with the columns: the left column of the square receives the colors of the column 4-hop right. The right column of the square receives the colors of the column 4-hop left. We have completed the coloring without using additional colors. This coloring uses exactly 16 colors, as depicted in Figure~\ref{ColorNeighbor3hopwith1.5Fig}.
\item \textit{Third step: This coloring can be regularly reproduced to constitute a valid coloring of the grid.} We consider the origin at node $u$. 
Observing the coloring depicted in Figure~\ref{ColorNeighbor3hopwith1.5Fig}, we notice that any color found at coordinates $(x,y)$ is also found:
\begin{itemize}
\item in the same line, at nodes $(x+4,y)$ and $(x-4,y)$,
\item in the same column, at nodes $(x,y+4)$ and $(x,y-4)$,
\item in the same diagonal, at nodes $(x+4, y-4)$, $(x-4,y-4)$, $(x-2,y+2)$ and $(x-2, y-2)$.
\end{itemize}
We then get a coloring of the grid with exactly 16 colors. We prove that this coloring is valid by checking that any color $1$ to $16$ is reused neither 1-hop, nor 2-hop, nor 3-hop away. It follows that this coloring is valid. Hence, an optimal coloring requires exactly 16 colors.
\end{enumerate}
\endproof

\begin{property}
A basic color pattern of the grid with a transmission range equal to 1.5 times the grid step is given by: 
\begin{verbatim}
   10 11 12 13
   14  5  4  3
   15  6  1  2 
   16  7  8  9 
\end{verbatim}
\end{property}
\proof
We can extract from the coloring of the grid a basic color pattern containing exactly 16 colors. This pattern is periodically reproduced to generate the coloring of the grid. Each color is reproduced according to the rules given previously. 
\endproof

\subsubsection{Transmission range = 2 x grid step}
%----------------------------------------

Assuming a transmission range equal to 2 grid units, we notice that each node $u$ that is not a border one has exactly at the communication sense:
\begin{itemize}
\item 12 1-hop neighbors: such nodes belong to the 1-diamond or 2-diamond, totalizing 4+8=12 nodes. 
\item 28 2-hop neighbors: such nodes belong to the 3-diamond or 4-diamond, totalizing 12+16=28 nodes.
\item 44 3-hop neighbors: such nodes belong to the 5-diamond or 6-diamond, totalizing 20+24=44 nodes.
\end{itemize}
For any non-border node $u$, its neighborhood up to 3-hop, $\mathcal{N}(u)$ is illustrated in Figure~\ref{Neighbor3hopwith2Fig}.
\begin{figure}[!h]
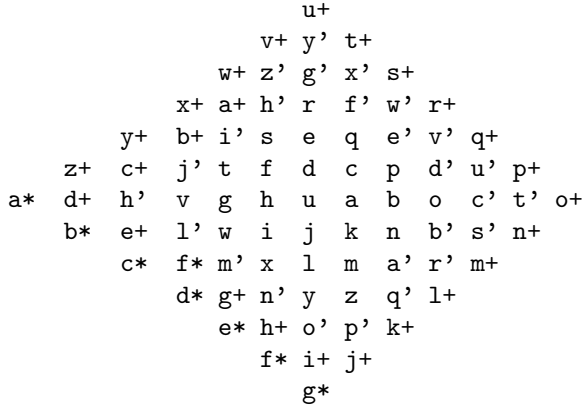

\begin{verbatim}
                     u+
                  v+ y' t+
               w+ z' g' x' s+
            x+ a+ h' r  f' w' r+
        y+  b+ i' s  e  q  e' v' q+
    z+  c+  j' t  f  d  c  p  d' u' p+
a*  d+  h'  v  g  h  u  a  b  o  c' t' o+ 
    b*  e+  l' w  i  j  k  n  b' s' n+ 
        c*  f* m' x  l  m  a' r' m+
            d* g+ n' y  z  q' l+
               e* h+ o' p' k+
                  f* i+ j+
                     g*
\end{verbatim}
\caption{Neighborhood up to 3-hop of node $u$, $R=2$.}
\label{Neighbor3hopwith2Fig}
\end{figure}

\begin{theorem}
An optimal coloring of a grid with a transmission range twice the grid unit needs exactly 25 colors. An example of optimal coloring is given by the following color pattern:
\begin{figure}[!h]
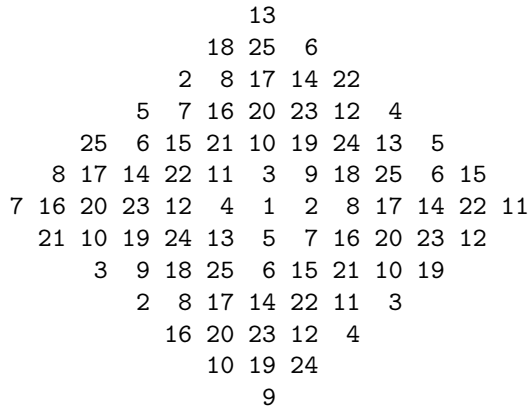

\begin{verbatim}
                  13
               18 25  6
             2  8 17 14 22
          5  7 16 20 23 12  4
      25  6 15 21 10 19 24 13  5
    8 17 14 22 11  3  9 18 25  6 15
 7 16 20 23 12  4  1  2  8 17 14 22 11
   21 10 19 24 13  5  7 16 20 23 12
       3  9 18 25  6 15 21 10 19
          2  8 17 14 22 11  3
            16 20 23 12  4
               10 19 24
                   9
\end{verbatim}
\caption{Coloring of $u$ and its neighborhood up to 3-hop, R=2.}
\label{ColorNeighbor3hopwith2Fig}
\end{figure}
\end{theorem}
\proof
Let $u$ be any non-border node. The proof is done in three steps:
\begin{enumerate}
\item \textit{First step: At least 25 colors are needed to color node $u$ and $\mathcal{N}(u)$.}\\
\textit{First substep:} node $u$ itself requires a color, denoted $1$ for simplicity reasons, that is not used by any other node in $\mathcal{N}(u)$.\\ 
\textit{Second substep:} any 1-hop neighbor of $u$ is 2-hop neighbor of any other 1-hop neighbor of $u$. It follows that any 1-hop neighbor requires a distinct color. We then get 12 colors, denoted $2$ to $13$, for these 1-hop neighbors.\\ 
\textit{Third substep:} we now consider the 2-hop neighbors. Notice that they are at most 3-hop away of any 1-hop neighbor of $u$. Hence, they cannot reuse the colors $2$ to $13$. Moreover, nodes that are on the 3-diamond, are at most 3-hop away from each other, they cannot reuse their colors. Hence, they need 12 additional colors: colors $14$ to $25$. It follows that at least 25 colors are needed to colour $u$ and $\mathcal{N}(u)$. 
\item \textit{Second step: We build a valid coloring of $u$ and $\mathcal{N}(u)$ with 25 colors.} 
We now consider the 2-hop neighbors of $u$ that belong to the 4-diamond. The upper-left side of this diamond can be colored by reproducing the colors of the lower right side of the 3-diamond, 4-hop away. Similarly, the lower right side of the 4-diamond can be colored by reproducing the colors of the upper-left side of 3-diamond. We color the upper-right side of the 4-diamond with the colors of the lower left side of the 3-diamond and the lower left side of the 4-diamond with the colors of the upper-right side of the 4-diamond.
We have colored node $u$ and all its nodes up to 2-hop with exactly 25 colors.\\

Concerning the 3-hop neighbors, they occupy the 5-diamond and 6-diamond. We first color the 5-diamond as follows: for its upper-left side, we reproduce the color of the lower-right side of the 2-diamond, similarly with the lower-right side reproducing the upper-left side of the 2-diamond. We proceed similarly for the upper-right and lower-left sides. We can now color the 6-diamond by reproducing colors used in the diagonals including the sides of the 1-diamond. We then obtain a valid coloring with exactly 25 colors, as depicted in Figure~\ref{ColorNeighbor3hopwith2Fig}.
\item \textit{Third step: This coloring can be regularly reproduced to constitute a valid coloring of the grid.} We consider the origin at node $u$. 
Observing the coloring depicted in Figure~\ref{ColorNeighbor3hopwith2Fig}, we notice that any color found at coordinates $(x,y)$ is also found at nodes $(x-4, y-3)$, $(x+3, y-4)$, $(x-3, y+4)$ and $(x+4, y+3)$.\\
We then get a coloring of the grid with exactly 25 colors. We prove that this coloring is valid by checking that any color $1$ to $25$ is reused neither 1-hop, nor 2-hop, nor 3-hop away. It follows that this coloring is valid. Hence, an optimal coloring requires exactly 25 colors.
\end{enumerate}
\endproof

\begin{property}
A basic color pattern of the grid with a transmission range equal to 2 times the grid step is given by: 
\begin{verbatim}
         20
      21 10 19
   22 11  3  9 18
23 12  4  1  2  8 17
   24 13  5  7 16
      25  6 15
         14 
\end{verbatim}
\end{property}
\proof
We can extract from the coloring of the grid a basic color pattern containing exactly 25 colors. This pattern is periodically reproduced to generate the coloring of the grid. Each color is reproduced according to the rules given previously. 
\endproof

\subsection{Optimal coloring for any transmission range: the Vector Method}
%-------------------------------------------------------------------------
In this section, we want to determine the optimal color number for the $h$-hop coloring of grids with various transmission ranges. 
As previously said, we only consider grid colorings that periodically reproduce a color pattern. 
We consider any node $U$ in the grid as depicted in Figure~\ref{figVM}. By definition of the $h$-hop coloring problem, the color of any node $U$ can be used by any node $V$ if and only if $V$ is more than $h$-hop away. We consider two nodes $V_1$ and $V_2$ that reuse the color of $U$. The parallelogram defined by nodes $U$, $V_1$ and $V_2$ constitutes a color pattern such that no node within this parallelogram reuse the color of $U$. This color pattern, periodically reproduced, must provide a valid coloring (no two nodes up to $h$-hop have the same color).
In order to optimize the spatial color reuse in the grid, the area of the parallelogram defined by $V_1$ and $V_2$ must be the smallest possible. The couples $(U,V_1)$ and $(U,V_2)$ determine two vectors that if independent generate the parallelogram of the color pattern. Hence the name of vector method. We now present this method more in details.\\

\begin{figure}[!h]
\label{figVM}
\begin{center}
{\includegraphics[width=0.6\linewidth]{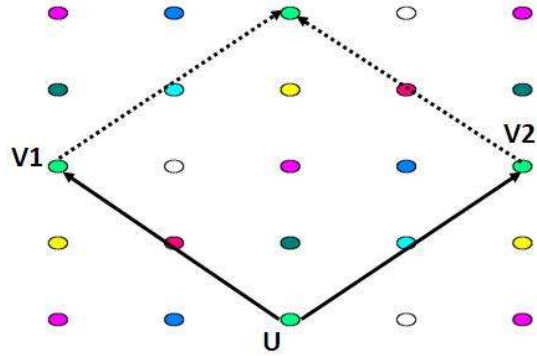}}
\caption{The principles of the vector method}\label{vectorMethodFig}
\end{center}
\end{figure}

We can notice that if $(U,V_1)$ and $(U,V_2)$ are generator vectors then any two vectors $(U,V'_1)$ and $(U,V'_2)$ that are a linear combination of $(U,V_1)$ and $(U,V_2)$ and are not dependant are also generator vectors. That is why, we can consider only the half plane delimited by $y \geq 0$.\\ 

\begin{definition}
With an optimal $h$-hop periodic coloring, the color of the origin node $U$ is reproduced at nodes $V_1$ and $V_2$ with coordinates $(x_1,y_1)$ and $(x_2,y_2)$ 
minimizing the determinant:\\
\begin{equation}
\label{equationminimize}
\mid x_1y_2 -x_2y_1 \mid
%\left| \begin{array}{lr} x_1  ~x_2 \\ y_1 ~y_2 \end{array} \right|
\end{equation}
under the constraints:\\
\begin{System}
y_1 \geq 0\\
y_2 \geq 0\\
x_1y_2-x_2y_1 \neq 0\\
V_1 \notin \bigcup_{k=1}^h k-hop(U)\\
V_2 \notin \bigcup_{k=1}^h k-hop(U)\\
V_2 \notin \bigcup_{k=1}^h k-hop(V_1)
\label{equationconstraint}
\end{System}
Vectors $(U,V_1)$ and $(U,V_2)$ are the two vectors used to generate the periodic color pattern of the $h$-hop coloring.
\end{definition}
System~\ref{equationminimize} means that $V_1$ and $V_2$ generate a parallelogram of the smallest possible area meeting the constraints of System~\ref{equationconstraint}. The first two constraints of System~\ref{equationconstraint} $y1 \geq 0$ and $y2 \geq 0$ show that we restrict our study to the half plane $y\geq0$. The third constraint $x_1y_2-x_2y_1 \neq 0$ expresses that the two vectors $(U,V_1)$ and $(U,V_2)$ are independent. The fourth constraint $V_1 \notin \bigcup_{k=1}^h k-hop(U)$ expresses that $V_1$ can reuse the color of $U$. Idem for the fifth constraint with $V_2$. The last constraint expresses that $V_1$ and $V_2$ can use the same color.

\begin{lemma}
\label{Noothernode}
No other node strictly within the parallelogram defined by $U$, $V_1$ and $V_2$ uses the same color as $U$.
\end{lemma}
\proof By contradiction, let us assume there exists $W$ a node strictly within the parallelogram defined by $U$,$V_1$ and $V_2$ that uses the same color as $U$. We distinguish two cases:
\begin{itemize}
\item $W$ is strictly more than $h$ hops away from $U$, $V_1$, $V_2$ and $V_3$ the fourth vertex of the parallelogram. In such conditions, the vectors $(U,V_1)$ and $(U,W)$ would form a parallelogram whose area is strictly smaller than this defined by $(U,V_1)$ and $(U,V_2)$. A contradiction with the generator vectors.
\item $W$ is at most $h$ hops away from at least one node among $U$, $V_1$, $V_2$ and $V_3$. Hence, $W$ cannot reuse the color: a contradiction of $h$-hop coloring. 
\end{itemize}
\endproof

\begin{property}
\label{numberofcolors}
For any node $U$, the color pattern defined by nodes $U$, $V_1$ and $V_2$ meeting the equation system \ref{equationminimize} under the constraints \ref{equationconstraint}, is periodic and contains exactly $\mid x_1y_2-x_2y_1 \mid$ colors.
\end{property}
\proof According to Lemma~\ref{Noothernode}, the color of the origin node $U$ is used only at nodes $W$, such that $(U,W)$ is a linear combination of $(U,V_1)$ and $(U,V_2)$. The coordinates $(x,y)$ of $W$ verify:\\
\begin{System}
x = \alpha x_1 + \beta x_2\\ 
y = \alpha y_1 + \beta y_2
\label{eq:alpha-beta}
\end{System}
where $\alpha \in Z$ and $\beta \in Z$. With these vectors, we can define a new grid, whose axis are given by $(U,V_1)$ and $(U,V_2)$, such that only nodes in this grid have the color of $U$.
Hence, the density of $U$'s color is equal to $1/d$, with $d=\mid x_1y_2 -x_2y_1 \mid$. Hence, we need $d$ colors to color all nodes.
\endproof

\begin{property}
\label{mycolor}
According to the color pattern defined by vectors $(U,V_1)$ and $(U,V_2)$ with $d= \mid x_1y_2-x_2y_1 \mid$, the color of any node $W(x,y)$ is determined by the couple $(c_1,c_2)$:
\begin{System}
\label{colorcompute}
c_1= \mid xy_1-yx_1 \mid\ modulo\ d\\
c_2= \mid xy_2-yx_2 \mid\ modulo\ d
\end{System}
There is a bijective mapping between the couple $(c_1,c_2)$ and a color $\in [1,d]$.
\end{property}
\proof
This property is deduced from property~\ref{numberofcolors}, by solving 
for $\alpha$ and $\beta$ in equation~\ref{eq:alpha-beta}
and then expressing the constraint that they must be integers.
\endproof

\subsection{Examples of vectors}
%-------------------------------

Table~\ref{TableVectorGrid} gives for different radio ranges two vectors generating the optimal periodic pattern as well as the minimal number of colors obtained by a periodic pattern, for both a 2-hop coloring and a 3-hop coloring. The '*' symbol highlights the optimality of the number of colors used.
\begin{table}[!h]
\caption{Vectors generating the optimal periodic pattern and optimal number of colors.}
\label{TableVectorGrid}
\begin{center}
\begin{tabular}{|c||c|c|c||c|c|c|}
\hline
Radio & \multicolumn{3}{|c||}{2-hop coloring}& \multicolumn{3}{|c|}{3-hop coloring}\\
           \cline{2-7}
range& vector1& vector2& colors& vector1& vector2& colors\\
\hline \hline
1& (2,1) & (-1,2)& 5*&(2,2) & (-2,2)& 8*\\
\hline
1.5& (-3,0)&(3,0)& 9*&(4,0) & (0,4)& 16*\\
\hline
2& (3,2)& (-2,3)& 13* &(4,3) & (-3,4)& 25*\\
\hline
2.5&(4,3)& (-1,5)& 23* &(5,5) & (-7,2)& 45*\\
\hline
3& (5,3)& (-1,6)& 33* &(7,5) & (-8,4)& 68*\\
\hline
3.5&(5,4)& (-6,3)& 39*& (8,5) & (-8,5)& 80*\\
\hline
4& (7,3)& (66,5)& 53*&(8,8) & (-11,3)& 112*\\
\hline
4.5& (9,2) & (-6,7)& 75*&(13,3) & (-9,10)& 157*\\
\hline
5& (9,4)& (-1,10) & 94*&(14,4) & (3,15)& 198*\\
\hline
5.5& (9,6)& (-1,11)&105*&(16,0) & (8,14)& 224*\\
\hline
6&(11,4)& (-9,8)&124* & (17,4)& (-12,13)& 269*\\
\hline
6.5&(13,1)& (-7,11)&150*&(-19,0)&(9,17)& 323*\\
\hline
7& (10,9)&(-4,13)& 166*&(15,13) &(-19, 7)& 352*\\
\hline
\end{tabular}
\end{center}
\end{table}

\subsection{Reduction of the number of vectors to test}
%---------------------------------------------
We can use some properties of the grid topology to reduce the number of vector computations. We need to characterize the set of neighbor nodes of the grid node $U$ more precisely.
\begin{property}
\label{1hopprop}
For any transmission range $R \geq 1$, for any grid node $U$, any node $V$ such that its euclidian distance $d(U,V)$ meets $0<d(U,V)\leq R$ belongs to 1-hop$(U)$.
\end{property}
\proof By definition of the radio range and assuming an ideal environment with symmetric links.
\endproof
\begin{lemma}
\label{distmaxg}
For any integer $h\geq1$, for any node $V$ such $d(U,V) \leq hR$, there exists a node $V'$ in the grid such that $d(V,V') \leq \sqrt{2}/2$.
\end{lemma}
\proof
Let us consider any node $V$ such that $d(U,V) \leq hR$.
In the worst case, this node occupies the center of the grid cell. It is at equal distance of two grid nodes that are diagonally opposed. Hence, its distance to one of them is equal to $\sqrt{2}/2$.
\endproof
\begin{property}
\label{2hopprop}
For any transmission range $R \geq 1$, for any node $U$, any node $V$ such that $R<d(U,V)\leq 2R$ belongs to either 2-hop$(U)$ or 3-hop$(U)$.
\end{property}
\proof For any node $V$ such that $R<d(U,V)\leq 2R$, we distinguish 2 cases:
\begin{itemize}
\item first case: $\exists W \in 1-hop(U)$ such that $d(V,W) \leq R$. Since $W \notin 1-hop(U)$, it follows that $V \in 2-hop(U)$.
\item second case: any 1-hop neighbor of $U$ is at a distance strictly higher than $R$ from $V$. Let $W$ be the 1-hop neighbor of $U$ the closest to $V$. We have $d(V,W) > R$, by definition. Let $Z$ be the 1-hop neighbor of $W$ the closest to $V$. We have $d(Z,V) \leq R$. Hence, $U$ reaches $V$ in three hops via nodes $W$ and $Z$.
\end{itemize}
\endproof

\begin{property}
\label{hhopprop}
For any grid node $U$, any node $V$ that meets $d(U,V) \leq (R-\sqrt{2})h$ is at most $h$-hop away from $U$.
\end{property}
\proof
We define the $h-1$ nodes that allow us to divide the distance $d(U,V)$ in $h$ equal parts.\\ Let $W_i$ be these nodes, with $i \in [1,h-1]$.\\ For any $i \in [1,h-1]$, let $W'_i$ the grid point the closest to $W_i$. For simplicity reason, we denote $W'_0=U$ and $W'_h=V$. We have $d(U,V) \leq \sum_{i=0}^{h-1} d(W'_i,W'_{i+1})$.\\ 
We have  $d(W'_i,W'_{i+1})\leq d(W'_i,W_i) + d(W_i,W_{i+1}) + d(W_{i+1},W'_{i+1})$.
According to Lemma~\ref{distmaxg}, we have $d(W_i, W'_i) \leq \sqrt{2}/2$. Hence, we get $d((W'_i,W'_{i+1})\leq \sqrt{2}+ d(W_i,W_{i+1})$. By construction, $d(W_i,W_{i+1})=d(U,V)/h$. Hence, $d((W'_i,W'_{i+1})\leq \sqrt{2}+ d(U,V)/h$.\\ 
If $\sqrt{2}+ d(U,V)/h \leq R$ then $d(U,V) \leq (R-\sqrt{2})h$.
Hence, $V$ is at most $h$-hop away from $U$.
\endproof

According to Property \ref{hhopprop}, we can notice that within the annulus delimited by the disks entered at $U$ and of radius $h(R-\sqrt{2})$ and $hR$, there exist nodes that are not $h$-hop nodes of $U$. These nodes are good candidates for $V_1$ and $V_2$. Hence, an heuristic is to reduce the set of possible solutions for $V_1$ and $V_2$ to only nodes $V$ that meet $h(R-\sqrt{2})<d(U,V)\leq hR$ and do not belong to the $h$-hop neighborhood of $U$. The new problem becomes:
\begin{definition}
Use a $h$-hop periodic coloring, where the color of the origin node $U$ is reproduced at nodes $V_1$ and $V_2$ with coordinates $(x_1,y_1)$ and $(x_2,y_2)$ 
minimizing:\\
\begin{equation}
\label{newequationminimize}
\mid x_1y_2 -x2y_1 \mid
\end{equation}
under the constraints:\\
\begin{System}
\label{newequationconstraint}
y_1 \geq 0\\
y_2 \geq 0\\
x_1y_2-x_2y_1 \neq 0\\
h^2(R-\sqrt{2})^2 < x_1^2+y_1^2 \leq h^2R^2\\
h^2(R-\sqrt{2})^2 < x_2^2+y_2^2 \leq h^2R^2\\
V_1 \notin h-hop(U)\\
V_2 \notin h-hop(U)\\
V_2 \notin \bigcup_{k=1}^h k-hop(V_1)
\end{System}
\end{definition}

\subsection{How to apply the Vector Method}
%---------------------------------------------
The Vector Method allows us to determine the optimal $h$-hop color pattern of any grid, with any transmission range. We can notice that any permutation of an optimal color pattern is still an optimal one. It follows that we can color each node within the parallelogram defined by the two generator vectors according to for example the line order within this parallelogram. We then get an optimal periodic color pattern.\\

More precisely, each node proceeds as follows:\\
1. Each node in the grid computes the two generator vectors. It is also possible that a central unit computes the two generator vectors and distributes this information to all nodes in the grid.\\ 
2. Each node colors each grid node in the parallelogram defined by the two generator vectors, following for instance the line order.\\ 
3. Knowing its coordinates in the grid, each node deduces the two components $c_1$ and $c_2$ according to property~\ref{mycolor}. It then deduces its own color from the color assigned to the node within the parallelogram with the same values of $c_1$ and $c_2$.\\

\begin{property}
The Vector Method provides the optimal number of colors for a periodic $h$-hop coloring of any grid, with any transmission range. It allows each node to know its color in a single round.
\end{property}

\subsection{Bounds of the number of colors in periodic colorings}
%--------------------------------------------------

In this section, we give a lower and an upper bound of the number of colors
needed in a $h$-hop coloring of the grid under the previously given conditions.
The bounds apply to the vector method, or any other method.
Combined in theorem~\ref{th:asymptotic}, the number of colors of 
optimal coloring when $R \rightarrow \infty$ is shown to be asymptotically
$\frac{\sqrt{3}}{2}h^2R^2 + O(R)$. 

Notice that this compares to a periodic coloring of a true hexagonal lattice 
which would yield a number of colors equal to $\frac{\sqrt{3}}{2}h^2R^2$,
and this is the best possible even when not constraining the nodes to be
located on a grid (see the circle packing in the proof of the next theorem).

For the lower bound, we have the following theorem:
\begin{theorem}
The number of colors required to color an infinite grid is at least 
$\frac{\sqrt{3}}{2}h^2(R-\sqrt{2})^2$
\end{theorem}

\begin{proof}
Consider $h$-hop coloring of the grid. 
Consider a fixed color $c$, and now let $S_c$ be the set of nodes with
this color. 

We first establish a lower bound of the distance of nodes
in $S_c$. 
Let us define $\rho = (R-\sqrt{2})h$.
Consider two nodes $A,B$ of $S_c$. By contradiction: if their distance
verified $d(A,B) \leq \rho$, from property~\ref{hhopprop},
they would be at most $h$-hop away, contradicting the definition of a 
$h$-hop coloring. 
Therefore, all nodes of $S$ are at a distance at least $\rho$ from
each other. 

Now consider the set of circles ${\cal C}$ of 
radius $\frac{1}{2}\rho$
and whose centers are the nodes of $S$. The fact that any two of nodes 
of $S$ are distant of more $\rho$, implies that none of the circles 
in ${\cal C}$ overlap. Hence ${\cal C}$ is a \emph{circle packing} 
by definition. From the Thue-T\'oth theorem 
\cite{t1910,t1943} establishing that the hexagonal circle packing 
is the densest packing, with a density of $\frac{\pi}{\sqrt{12}}$, 
we deduce that ${\cal C}$
must have a lower or equal packing density. 
This implies an upper bound of the density of set $S_c$ of centers of the disks
of $\frac{1}{(\rho/2)^2\sqrt{12}}$.

Because every color yields a set of nodes with at most this density, it
follows a lower bound of the number of colors that is the inverse of this
quantity, hence the theorem.

\end{proof}

For an upper bound, we have the following theorem:
\begin{theorem}
The number of colors required to color an infinite grid is at most
$\frac{\sqrt{3}}{2}h^2R^2 + 2 hR + (2+hR) \sqrt{2}$
\end{theorem}

\begin{proof}

We proceed with a constructive proof, exhibiting two valid vectors which yield 
the result, using an approximation of an hexagonal lattice.

The figure~\ref{fig:hexagon-choice} illustrates how some points $V_2$ and
$V_1$ are constructed. 
\begin{figure}[!h]
\begin{center}
{\includegraphics[width=0.6\linewidth]{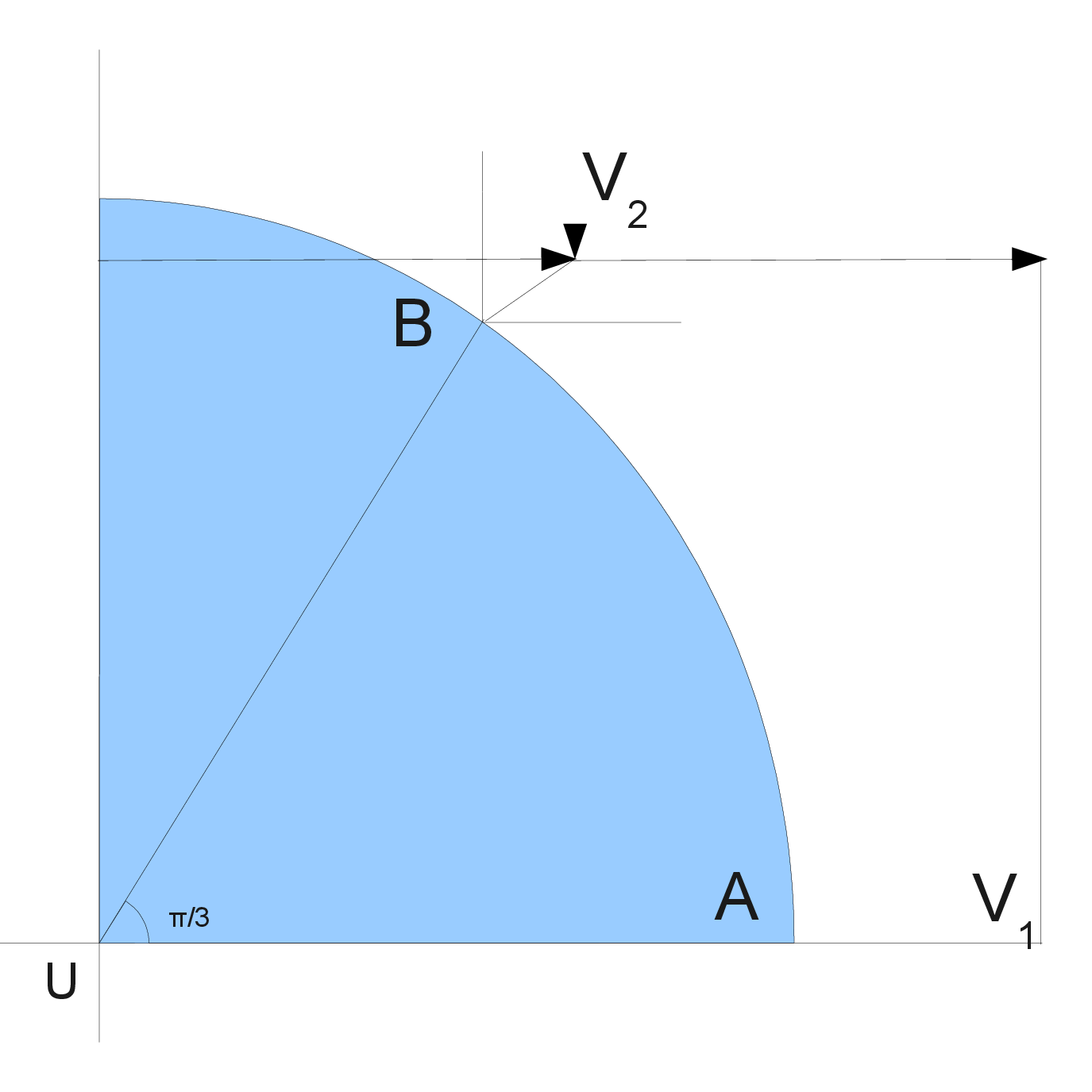}}
\caption{Selecting vectors for a near-hexagonal lattice}\label{fig:hexagon-choice}
\end{center}
\end{figure}
\begin{itemize}
\item Starting from the point $U$, the line with an 
angle $\frac{\pi}{3}$ with the horizontal line is considered,
and its intersection with the circle of radius $h R$ yields the point $B$.
\item Next, the closest point of $B$ on the grid
  with larger $x$ and also $y$ coordinates, is sought and is $V_2$
  (coordinates $(x_2, y_2)$). 
\item Then $V_1$ with coordinates $(x_1, y_1)$ is selected with
  $(x_1,y_1) = (2x_2, 0)$.
\end{itemize}

Notice that by construction $x_1 \ge h R$, and we have a valid choice of
vectors $(U,V_1)$ and $(U,V_2)$.

By construction:
$d(B,V_2) \le \sqrt{2}$ and $d(A,V_1)~\le~2$.

Using the general notations $MN$ to represent a vector $(M,N)$, 
and $det(OM,ON)$ to represent the determinant of two vectors,
we can write $n_c$, the number of the color in the associated coloring as:
\begin{eqnarray*}
n_c & = & det(UV_1,UV_2) \\
    & = & det(UA,UB) + det(AV_1,UB) + det(UV_1,BV_2) \\
    & \le & det(UA,UB) + d(A,V_1)d(U,B) + d(U,V_1)d(B,V_2) \\
    & \le & \frac{\sqrt{3}}{2}h^2R^2 + 2 hR + (2+hR) \sqrt{2}
\end{eqnarray*}

\end{proof}

The two previous theorems show that the number of colors
of an optimal periodic $h$-hop coloring of the grid
(with the vector method) is close the number of colors in an absolutely optimal
coloring of an hexagonal lattice,
at least asymptotically. This result may be summarized as:
\begin{theorem}
The number of colors $n_c(R) $of an optimal periodic $h$-coloring  for a fixed
verifies:
$$n_c(R) = \frac{\sqrt{3}}{2}h^2R^2 (1 + O(\frac{1}{R}))$$
when $R \rightarrow \infty$
\label{th:asymptotic}
\end{theorem}
\begin{proof}
Combining the lower bound and the upper bound of the two theorems
yields the result.
\end{proof}

\subsection{Coloring results with the Vector Method}
%================================================
Table~\ref{TableVMVariousPrio} depicts the results obtained with the vector method for various grids, with various radio ranges. The priority of any node is given by the couple $c_1,c_2)$ as defined in property~\ref{mycolor}.
Results are given for 3-hop coloring. The '*' symbol highlights the optimality of the number of colors used.

\begin{table}[!h]
\caption{Number of colors obtained for 3-hop coloring.}
\label{TableVMVariousPrio}
\begin{center}
\begin{tabular}{|c|c|c|c|}
\hline
Radio range R& Grid size& priority assignment& colors\\
\hline \hline
1&10x10& line& 8* \\
\cline{3-4}
 && column & 8*\\
\cline{3-4}
 && vector & 8*\\
\cline{2-4}
&20x20& vector& 8*\\
\cline{2-4}
&30x30& vector& 8*\\
\cline{2-4}
&50x50& vector& 8*\\
\hline
1.5&10x10& line & 16*\\
\cline{3-4}
 && column & 16*\\
\cline{3-4}
 && vector & 16*\\
 \cline{2-4}
&20x20& vector& 16*\\
\cline{2-4}
&30x30& vector& 16*\\
\cline{2-4}
&50x50& vector& 16*\\
\hline
2&10x10& line & 30\\
\cline{3-4}
 && column & 30\\
\cline{3-4}
 && vector & 25*\\
 \cline{2-4}
&20x20& vector& 25*\\
\cline{2-4}
&30x30& vector& 25*\\
\cline{2-4}
&50x50& vector& 25*\\
\hline
3& 20x20& vector & 68*\\
\cline{2-4}
&30x30& vector& 68*\\
\hline
\end{tabular}
\end{center}
\end{table}
We observe that the vector method provides an optimal three-hop coloring, for any radio range. This is not true for any other priority assignment tested. Moreover, the number of colors does not depend on the grid size. Similar results have been obtained for 2-hop coloring.%Furthermore, the impact of the grid size on the number of rounds is very limited.

\section{Improvement of SERENA with theoretical results \label{Improvement}}
%==================================================================
Our goal is now to optimize the number of colors obtained by SERENA using the theoretical results obtained in the previous section. We will act on node priority assignment. How can each node compute its priority in order to minimize the number of colors used? 
We now show how to compute node priority to allow SERENA to obtain an optimal coloring of a grid.
We try different node priority assignments:
line, column, diagonal, distance to the grid center and vector.
Table~\ref{TableSerenaVariousPrio} depicts the simulation results obtained with SERENA for various grids with various radio ranges.
The '*' symbol highlights the optimality of the number of colors used.

\begin{table}[!h]
\caption{Number of colors obtained with SERENA.}
\label{TableSerenaVariousPrio}
\begin{tabular}{|c|c|c|c|c|}
\hline
Radio range R& Grid size& priority assignment& colors & rounds\\
\hline
1&10x10& line& 8* & 58\\
\cline{3-5}
 && column & 8* & 58\\
\cline{3-5}
 && vector & 8* & 21\\
\cline{2-5}
&20x20& vector& 8*& 21\\
\cline{2-5}
&30x30& vector& 8*& 21\\
\cline{2-5}
&50x50& vector& 8*& 21\\
\hline
1.5&10x10& line & 16*& 91\\
\cline{3-5}
 && column & 16* & 91\\
\cline{3-5}
 && vector & 16* & 38\\
 \cline{2-5}
&20x20& vector& 16*& 38\\
\cline{2-5}
&30x30& vector& 16*& 38\\
\cline{2-5}
&50x50& vector& 16*& 38\\
\hline
2&10x10& line & 30 & 85\\
\cline{3-5}
 && column & 30 & 85\\
\cline{3-5}
 && vector & 25* & 52\\
 \cline{2-5}
&20x20& vector& 25*& 56\\
\cline{2-5}
&30x30& vector& 25*& 61\\
\cline{2-5}
&50x50& vector& 25*& 68\\
\hline
3& 20x20& vector & 68* & 179\\
\cline{2-5}
&30x30& vector& 68*& 184\\
\hline
\end{tabular}
\end{table}
We observe that SERENA provides an optimal three-hop coloring with the priority assignment based on vectors, for any radio range. This is not true for any other priority assignment tested. Moreover, the number of colors does not depend on the grid size. Furthermore, the impact of the grid size on the number of rounds is very limited.

\section{Conclusion\label{Conclusion}}
%=====================================
In this paper we have proved complexity of the $h$-hop node coloring problem. We have then optimized SERENA a 3-hop node coloring algorithm for dense networks without sacrificing the coloring delay. We have then focused on specific case of dense networks: grids with a radio range higher than the grid step. We have established theoretical results about grid coloring. We have proposed the Vector Method for assigning colors to nodes such as sensors organized in grid. We have also given lower and upper bounds of the number of colors in periodic colorings. As a further work, we will show how to map a grid on a given random topology and determine the best grid adapted to this topology.


\begin{thebibliography}{1}

%---------------------- Complexity ------------------------------
\bibitem{gar79}
Garey, M.; Johnson, D., 
{\em Computers and intractability: a guide to theory of NP-completeness}, W.H. Freeman, San Francisco, California, 1979.


%--------------------- Graph coloring --------------------------------
\bibitem{brel79}
Brelaz, D., 
{\em New methods to color the vertices of a graph}, Communications of the ACM, 22(4), 1979.

\bibitem{hans04}
Hansen, J.; Kubale, M.; Kuszner, L.; Nadolski, A.,
{\em Distributed largest-first algorithm for graph coloring}, EURO-PAR 2004, Pisa, Italy, August 2004.

\bibitem{kuhn06}
Kuhn F., Wattenhofer R.,
{\em On the complexity of distributed graph coloring}, PODC 2006, Denvers, Colorado, July 2006.

\bibitem{Gandham08}
Gandham, S.; Dawande, M.; Prakash, R.,
Link scheduling in sensor networks: distributed edge coloring revisited,
{\em Journal of Parallel and Distributed Computing}, vol. 68, 8, August 2008.

%---------------------------- Network coloring -----------------------------
\bibitem{Rama89}
Ramaswami, R.; Parhi, K.,
{\em Distributed scheduling of broadcasts in a radio network},
INFOCOM 1989,
Ottawa, Canada, April 1989.

\bibitem{Krumke00}
Krumke, S.; Marathe, M.; Ravi, S.,
Models and approximation algorithms for channel assignment in radio networks,
{\em Wireless Networks}, vol. 7, 2000.

\bibitem{interfcolor}
Jain, K.; Padhye, J.; Padmanabhan, V.; Qiu, L.,
{\em Impact of interference on multi-hop wireless network performance},
ACM MobiCom, San Diego, CA, 2003.

%------------------------- TDMA + slot assignment --------------------------------%
\bibitem{rajendran03}
Rajendran, V.; Obraczka, K.; Garcia-Luna-Aceves, J.-J.,
{\em Energy-efficient, collision-free medium access control for wireless sensor networks},
Sensys'03,
Los Angeles, California, November 2003.

\bibitem{DRAND}
Rhee, I.; Warrier, A.; Xu, L.,
{\em Randomized dining philosophers to TDMA scheduling in wireless sensor networks},
Technical Report TR-2005-21, Dept of Computer Science,
North Carolina State University, April 2005.

\bibitem{ZMAC}
Rhee, I.; Warrier, A.; Aia, M.; Min, J.,
{\em {Z-MAC}: a hybrid {MAC} for wireless sensor networks},
SenSys'05, San Diego, California, November 2005.

\bibitem{rajendran05} 
Rajendran, V.; Garcia-Luna-Aceves, J.J.; Obraczka, K.,
{\em Energy-efficient, application-aware medium access for sensor networks},
IEEE MASS 2005,
Washington, November 2005.

\bibitem{TDMA-ASAP}
Gobriel, S.; Mosse, D.; Cleric, R.,
{\em TDMA-ASAP: sensor network TDMA scheduling with adaptive slot stealing and parallelism},
ICDCS 2009, Montreal, Canada, 
June 2009.

\bibitem{Ma09}
Ma, J.; Lou, W.; Wu, Y.; Li, X.-Y.,
{\em Energy efficient TDMA sleep scheduling in wireless sensor networks},
INFOCOM 2009,
Rio de Janeiro, Brazil, April 2009.

\bibitem{iwcmc08}
Minet, P.; Mahfoudh, S.,
{\em SERENA: SchEduling RoutEr Nodes Activity in wireless ad hoc and sensor networks}, IWCMC 2008, IEEE International Wireless Communications and Mobile Computing Conference, Crete Island, Greece, August 2008.

\bibitem{Chakra04}
Chakraborty, G.,
Genetic algorithm to solve optimum TDMA transmission schedule in broadcast packet radio networks,
{\em IEEE Transactions on Communications}, vol. 52, 5, May 2004.

\bibitem{FlexiTP}
Lee, W. L.; Datta A.; Cardell-Oliver, R.,
FlexiTP: a flexible-schedule-based TDMA protocol for fault-tolerant and energy-efficient wireless sensor networks,
{\em IEEE Transactions on Parallel and Distributed Systems}, vol. 19, 6, June 2008.

%------------------------------- Tight bound and first fit --------------------
\bibitem{Car07}
Cargiannis, I.; Fishkin A.; Kaklamanis C.; Papaioannou E.,
A tight bound for online coloring of disk graphs,
{\em Theoretical Computer Science}, 384, 2007.

%----------------------------Circle packing -----------------------------------
\bibitem{t1910}
A. Thue, {\em \"Uber die dichteste Zusammenstellung von kongruenten Kreisen in
einer Ebene}, Norske Vid. Selsk. Skr. No.1 (1910), 1-9.

\bibitem{t1943} L. F. T\'oth, {\em \"Uber die dichteste Kugellagerung},
Math. Z. 48 (1943), 676-684.

\end{thebibliography}
\end{document}